\documentclass[twoside,leqno]{article}

\usepackage[letterpaper]{geometry}

\usepackage{ltexpprt}
\usepackage[backref=page, colorlinks,
            urlcolor=black,
            linkcolor=blue,      
            anchorcolor=blue,  
            citecolor=blue,       
            ]{hyperref}

\begin{document}

\newcommand\relatedversion{}
\renewcommand\relatedversion{\thanks{The full version of the paper can be accessed at \protect\url{https://arxiv.org/abs/1902.09310}}} 

\title{\Large Faster Algorithms for Bounded Knapsack and Bounded Subset Sum Via Fine-Grained Proximity Results}
\author{
  Lin Chen\thanks{ chenlin198662@gmail.com. Texas Tech University. Partially supported by NSF CCF [No. 2004096]}
  \and 
  Jiayi Lian\thanks{jiayilian@zju.edu.cn. Zhejiang University.}
  \and
  Yuchen Mao\thanks{maoyc@zju.edu.cn. Zhejiang University. Partially supported by NSFC [No. 12271477]}
  \and
  Guochuan Zhang\thanks{zgc@zju.edu.cn. Zhejiang University. Partially supported by NSFC [No. 12131003]}
}

\date{}

\maketitle

\begin{abstract} \small\baselineskip=12pt We investigate pseudopolynomial-time algorithms for Bounded Knapsack and Bounded Subset Sum. Recent years have seen a growing interest in settling their fine-grained complexity with respect to various parameters. For Bounded Knapsack, the number of items $n$ and the maximum item weight $w_{\max}$ are two of the most natural parameters that have been studied extensively in the literature. The previous best running time in terms of $n$ and $w_{\max}$ is $O(n + w^3_{\max})$ [Polak, Rohwedder, W\k{e}grzycki '21]. 
There is a conditional lower bound of $(n + w_{\max})^{2-o(1)}$ based on $(\min,+)$-convolution hypothesis [Cygan, Mucha, Węgrzycki, Włodarczyk '17]. We narrow the gap significantly by proposing an $\widetilde{O}(n + w^{12/5}_{\max})$-time algorithm. Our algorithm works for both $0$-$1$ Knapsack and Bounded Knapsack. 
   Note that in the regime where $w_{\max} \approx n$, our algorithm runs in $\widetilde{O}(n^{12/5})$ time, while all the previous algorithms require $\Omega(n^3)$ time in the worst case.

   For Bounded Subset Sum, we give two algorithms running in $\widetilde{O}(nw_{\max})$ and $\widetilde{O}(n + w^{3/2}_{\max})$ time, respectively. These results match the currently best running time for 0-1 Subset Sum. Prior to our work, the best running times (in terms of $n$ and $w_{\max}$) for Bounded Subset Sum are $\widetilde{O}(n + w^{5/3}_{\max})$ [Polak, Rohwedder, Węgrzycki '21]  and $\widetilde{O}(n + \mu_{\max}^{1/2}w_{\max}^{3/2})$ [implied by Bringmann '19 and Bringmann, Wellnitz '21], where $\mu_{\max}$ refers to the maximum multiplicity of item weights.\end{abstract}

\section{Introduction}
Knapsack and Subset Sum are two of the most fundamental problems in combinatorial optimization.  In 0-1 Knapsack, we are given a set of $n$ items and a knapsack of capacity $t$.  Each item $i$ has a weight $w_i$ and a profit $p_i$.  Assume $t, p_i, w_i \in \mathbb{N}$.  We should select items so as to maximize the total profit subject to the capacity constraint.  Subset Sum is a special case of Knapsack where every item has its profit equal to its weight. In 0-1 Subset Sum, given a set of $n$ items with weight $\{w_i\}_{i \in {[n]}}$ and a target $t$, we are asked whether there is some subset of items whose total weight is $t$.  The two problems can be naturally generalized to the bounded case where each item $i$ can be selected up to $u_i$ times. We investigate Bounded Knapsack and Bounded Subset Sum. Throughout the rest of the paper, Knapsack and Subset Sum always refer to the bounded case unless stated otherwise.



Knapsack and Subset Sum are both weakly NP-hard, and can be solved in pseudopolynomial time using standard dynamic programming~\cite{Bel57}. In recent years, there has been a series of works trying to settle the best possible pseudopolynomial running time for these problems with respect to various parameters, including $n$, $t$, the total number of item copies $N = \sum_i u_i$, the maximum weight $w_{\max}$, the maximum profit $p_{\max}$, and the optimal total profit $OPT$. Table \ref{table:Knapsack} and Table \ref{table:Subset Sum} list the known pseudopolynomial-time algorithms for Knapsack and Subset Sum, respectively. Of particular interest are those algorithms that have only $n$ and $w_{\max}$ as parameters, because $n$ and $w_{\max}$ can be much smaller than $N$ and $t$, but not vice versa. 

\begin{table}[!ht]
    \centering
    \caption{Pseudopolynomial-time algorithms for Bounded Knapsack.}
    \begin{tabular}{cc}
        \toprule 
        Bounded Knapsack & Reference \\
        \hline \specialrule{0em}{0pt}{2pt}
        $\widetilde{O}(n\cdot\min\{t,OPT\})$ & Bellman \cite{Bel57} \\
        $O(N\cdot w_{\max}\cdot p_{\max})$ & Pisinger \cite{Pis99} \\
        $O(n^3\cdot w_{\max}^2)$ & Tamir \cite{Tam09} \\
        $\widetilde{O}(n+w_{\max}\cdot t)$ & Kellerer and Pferschy \cite{KP04}, also \cite{BHSS18,AT19} \\
        $\widetilde{O}(n+p_{\max}\cdot t)$ & Bateni et al. \cite{BHSS18} \\
        $\widetilde{O}(n\cdot w_{\max}^2\cdot\min\{n, w_{\max}\})$ & Bateni et al. \cite{BHSS18} \\
        $O(N\cdot \min\{w_{\max}^2,p_{\max}^2\})$ & Axiotis and Tzamos \cite{AT19} \\
        $\widetilde{O}(n\cdot\min\{w_{\max}^2,p_{\max}^2\})$ & Eisenbrand and Weismantel \cite{EW19} \\
        ${O}(n+\min\{w_{\max}^3,p_{\max}^3\})$ & Polak et al. \cite{PRW21} \\
        $\widetilde{O}(n+(t+OPT)^{1.5})$ & Bringmann and Cassis \cite{BC22} \\
        $\widetilde{O}(N\cdot \min\{w_{\max}\cdot p_{\max}^{2/3}, p_{\max}\cdot w_{\max}^{2/3}\})$ & Bringmann and Cassis \cite{BC23} \\
        $\widetilde{O}(n +\min\{w_{\max}^{12/5}, p_{\max}^{12/5}\})$ & This Paper \\
        \bottomrule 
    \end{tabular}
    \label{table:Knapsack}
\end{table}

\begin{table}[!ht]
    \centering
    \captionsetup{width=0.72\textwidth}
    \caption{Pseudopolynomial-time algorithms for Bounded Subset Sum. In (*), $\mu_{\max} = \max_{w}\{\sum_{i: w_i = w} u_i\}$.}
    \begin{tabular}{cc}
        \toprule 
        Bounded Subset Sum & Reference \\
        \hline \specialrule{0em}{0pt}{2pt}
        $\widetilde{O}(n\cdot\min\{t,OPT\})$ & Bellman \cite{Bel57} \\
        $O(N\cdot w_{\max})$ & Pisinger \cite{Pis99} \\
        $\widetilde{O}(n+\min\{\sqrt{n}\cdot t, t^{5/4}\})$ & Koiliaris and Xu \cite{KX19} \\
        $\widetilde{O}(n+t)$ & Bringmann \cite{Bri17}, Jin and Wu \cite{JW18} \\
        $\widetilde{O}(N+\mu_{\max}^{1/2}\cdot w_{\max}^{3/2})${*} & Bringmann and Wellnitz \cite{BW21} + \cite{Bri17}\\
        $\widetilde{O}(n+w_{\max}^{5/3})$ & Polak et al. \cite{PRW21}\\
        $\widetilde{O}(n\cdot w_{\max})$ & This Paper\\
        $\widetilde{O}(n+w_{\max}^{3/2})$ & This Paper\\
        \bottomrule 
    \end{tabular}
    \label{table:Subset Sum}
\end{table}

For Knapsack, the standard dynamic programming has a running time of $O(N^2w_{\max})$ as $t \leq Nw_{\max}$.  Tamir~\cite{Tam09} gave an $O(n^3w^2_{\max})$-time algorithm, which is the first algorithm with running time depending only on $n$ and $w_{\max}$. The results via fast $(\min,+)$-convolution in Reference~\cite{AT19,BHSS18,KP04} imply an $\widetilde{O}(Nw^2_{\max})$-time\footnote{In the paper we use an $\widetilde{O}(\cdot)$ notation to hide polylogarithmic factors.} algorithm.  Bateni et al.~\cite{BHSS18} also gave an improved $\widetilde{O}(n\cdot w_{\max}^2\cdot\min\{n, w_{\max}\})$-time algorithm. Eisenbrand and Weismantel~\cite{EW19} proved a proximity result, with which an $\widetilde{O}(nw^2_{\max})$-time algorithm for Knapsack can be obtained. Combining the proximity result~\cite{EW19} and the convolution framework~\cite{AT19,BHSS18,KP04}, Polak et al.~\cite{PRW21} showed that Knapsack can be solved in $O(n + w^3_{\max})$ time. There is also a conditional lower bound (for 0-1 Knapsack) of $(n + w_{\max})^{2-o(1)}$ based on the $(\min,+)$-convolution conjecture~\cite{CMWW19,KPS17}. It remains open whether Knapsack can be solved in $O(n + w^2_{\max})$ time~\cite{BC22}.

For Subset Sum, Pisinger~\cite{Pis99} improved the standard dynamic programming and obtained a running time of $O(N w_{\max})$. We remark that the dense Subset Sum result by Galil and Margalit~\cite{GM91} implies an $\widetilde{O}(n + w^{3/2}_{\max})$-time algorithm, but their algorithm requires all $w_i$ to be distinct and $u_{\max} = 1$ and is not applicable for Bounded Subset Sum. Hence, we do not include it in Table~\ref{table:Subset Sum}. Nevertheless, Bringmann and Wellnitz~\cite{BW21} later refined and extended the dense Subset Sum result, and their result implies an $\widetilde{O}(N + \mu_{\max}^{1/2}w^{3/2}_{\max})$-time algorithm where $\mu_{\max} = \max_w \sum_{i : w_i = w}u_i$. (Note that $\sum_{i : w_i = w}u_i$ is the total number of item copies whose weight is $w$). Polak et al.~\cite{PRW21} gave the first algorithm for Bounded Subset Sum that depends only on $n$ and $w_{\max}$. Their algorithm runs in $\widetilde{O}(n + w^{5/3}_{\max})$-time. There is a conditional lower bound of $\Omega(n + w^{1-o(1)}_{\max})$ implied by the Set Cover Hypothesis~\cite{Bri17} and Strong Exponential Time Hypothesis~\cite{ABHS22}.


\subsection{Our Results}
\begin{theorem}\label{thm:knapsack}
    There is an $\widetilde{O}(n+w_{\max}^{12/5})$-time algorithm for Bounded Knapsack.
\end{theorem}
Our result significantly narrows the gap between the previous upper bound of $O(n + w^3_{\max})$ and the lower bound of $\Omega(n + w_{\max}^{2-o(1)})$. Moreover, in the regime where $w_{\max}$ is $O(n)$, our algorithm outperforms all the previous algorithms that have only $n$ and $w_{\max}$ as parameters. In particular, when $w_{\max}\approx n$ and $t\approx n^2$, our algorithm is the first to guarantee a subcubic running time in this regime, while all previous algorithms require $\Omega(n^3)$ time\footnote{The only exception is the $\widetilde{O}(N w_{\max}p^{2/3}_{\max})$-time algorithm by Bringmann~\cite{BC23}, but they additionally require $p_{\max} \leq O(n^{3/2})$ and $N \approx n$.}.  

Due to the symmetry of weights and profits, we can also obtain an $\widetilde{O}(n + p^{12/5}_{\max})$-time algorithm by exchanging the roles of weights and profits~\cite{AT19,PRW21}.
\begin{corollary}\label{coro:knapsack}
    There is an $\widetilde{O}(n+p_{\max}^{12/5})$-time algorithm for Bounded Knapsack.
\end{corollary}

\begin{theorem}\label{thm:subset-sum-a}
    There is an $\widetilde{O}(nw_{\max})$-time algorithm for Bounded Subset Sum.
\end{theorem}

\begin{theorem}\label{thm:subset-sum-b}
    There is an $\widetilde{O}(n + w^{3/2}_{\max})$-time randomized algorithm for Bounded Subset Sum.
\end{theorem}

The best algorithms (in terms of $n$ and $w_{\max}$) for 0-1 Subset Sum run in $O(nw_{\max})$ and $\widetilde{O}(n + w^{3/2}_{\max})$ time, but they cannot be generalized to Bounded Subset Sum. Indeed, prior to our work, the best-known algorithm for Bounded Subset Sum in terms of these two parameters is that of Polak et al.~\cite{PRW21} running in $\widetilde{O}(n+w_{\max}^{5/3})$ time. Pisinger's algorithm~\cite{Pis99} runs in $O(n w_{\max})$ time on 0-1 Subset Sum but requires $O(Nw_{\max})$ on Bounded Subset Sum\footnote{Note that the idea of binary encoding each $u_i$ as $1 + 2 + 4 + \ldots +2^{\log u_i}$ (bundling item copies) does not help to reduce the running time to $O(n w_{\max})$, because although it reduces $u_i$ to a constant, it increases $w_{\max}$ by a factor of $u_{i}$.}. Meanwhile, the works of Galil and Margalit~\cite{GM91} and Bringmann and Wellnitz~\cite{BW21} imply a running time of $\widetilde{O}(n + w^{3/2}_{\max})$ only when all $\{w_1, \ldots, w_n\}$ are distinct and all $u_i = 1$.  For Bounded Subset Sum, their algorithms either do not work or have additional factors in the running time. It is natural to ask whether Bounded Subset Sum can be solved as fast as 0-1 Subset Sum.  Our results provide an affirmative answer to this question (at least with respect to the known algorithms).

\subsection{Technique Overview}
\paragraph{Knapsack} A proximity result~\cite{EW19} states that a certain greedy solution $\ve g$ and an optimal solution $\ve z$ differ in $O(w_{\max})$ items. It directly follows that the total weight $\Delta$ of these items is $O(w^2_{\max})$. Therefore, it is possible to obtain ${\ve z}$ from ${\veg}$ by deleting and adding only $O(w^2_{\max})$ volume of items.  
The upper bound on $\Delta$ plays an important role in accelerating the $(\min, +)$-convolution framework~\cite{AT19,BHSS18,KP04} for Knapsack. The $(\min, +)$-convolution-based algorithm first partitions items by their weights, and for each group, computes an all-capacities solution which is a vector of length $t+1$. These vectors are combined one by one using a linear-time concave $(\min, +)$-convolution algorithm~\cite{AKM+87}, and every intermediate vector is also truncated to be of length $t$. 
Since the concave $(\min, +)$-convolution are performed for at most $w_{\max}$ times, each taking $O(t)$ time, the total running time is $O(n+ w_{\max}t)$. Applying the proximity result, every solution vector and intermediate vector can be truncated to be of length roughly $w^2_{\max}$. Therefore, a running time of $O(n + w^3_{\max})$ can be obtained in~\cite{PRW21}.

The upper bound of $O(w^2_{\max})$ on $\Delta$ is tight in the sense that there are instances on which $\Delta = \Theta(w^2_{\max})$. Nevertheless, $\Delta$ only characterizes the total difference of ${\ve g}$ and ${\ve z}$ on all groups, and it does not provide satisfactory answers to questions like: \emph{how much can a particular group (or a particular set of groups) contribute to $\Delta$?} Intuitively, since the total contribution of all groups is $O(w^2_{\max})$, only a few groups can have large contribution. Our first fine-grained proximity result states that only $\widetilde{O}(w_{\max}^{1/2})$ groups can contribute as much as $O(w^2_{\max})$, and rest of the groups together can contribute at most $O(w^{3/2}_{\max})$. Then except for $\widetilde{O}(w_{\max}^{1/2})$ convolutions, every convolution can be done in $O(w^{3/2}_{\max})$ time, which leads to an $\widetilde{O}(n + w^{5/2}_{\max})$-time algorithm. Then we further improve the proximity result and obtain an $\widetilde{O}(n + w^{12/5}_{\max})$ algorithm.

We define the efficiency of an item to be the ratio of its profit to weight. Our proximity result is based on a natural intuition that only the choice of items with median efficiency can differ a lot in ${\ve g}$ and ${\ve z}$, while the items with high efficiency or low efficiency are so good or so bad that both ${\ve g}$ and ${\ve z}$ tend to make the same decision. Formal proof of this idea requires additive combinatorics tools developed by a series of works including \cite{Sar94, GM91, Lev03, Bri17}. Basically, when we select enough items with median efficiency, the multiset of weights of these items becomes ``dense''. Then we can use tools from dense subset sum due to Bringmann and Wellnitz~\cite{BW21} to show that items with high and low efficiency cannot differ too much in ${\ve g}$ and ${\ve z}$.

We remark that a recent work of Deng et al.~\cite{DJM23} also utilized the additive combinatorics result in Subset Sum~\cite{BW21} to show a proximity result in the context of approximation algorithms for Knapsack. Both our approach and theirs share a similar high-level idea, that is, if the optimal solution selects many low-efficiency items, then the high-efficiency items not selected by the optimal solution cannot form a dense set (in terms of their weights or profits), for otherwise it is possible to utilize additive combinatorics results to conduct an exchange argument on low- and high-efficiency items. There are, however, two major differences in techniques: (i). The density requirement on the high-efficiency item set in Deng et al.'s work~\cite{DJM23} is very strong. In particular, they require any large subset of this set to be dense. This works in approximation algorithms where the input instance can be reduced to a bounded instance in which item profits differ by a constant factor. For exact algorithms, such kind of density requirement is not achievable. Nevertheless, we give a similar exchange argument that only requires a weaker density condition.
 (ii). Multiplicity of item weights or profits has not been considered in Deng et al.'s work~\cite{DJM23}. We remark that the additive combinatorics result in Subset Sum by Bringmann and Wellnitz~\cite{BW21} is for multisets, which actually provides a trade-off between ``dense threshold'' and item multiplicity. Very roughly speaking, for a multiset $X$ that consists of integers up to $w$, it can become dense if it contains $\widetilde{\Theta}(\sqrt{w})$ or more distinct integers, but it can also become dense if it only contains $\widetilde{\Theta}(w^{0.4})$ distinct integers where each integer has $\widetilde{\Theta}(w^{0.2})$ multiplicities. Deng et al.'s work only utilized the former result. We achieve an $\widetilde{O}(n+w_{\max}^{5/2})$-time algorithm using the former result, and show that it can be further improved to $\widetilde{O}(n+w_{\max}^{12/5})$-time by exploiting the multiplicity to build a stronger proximity result.


\paragraph{Subset Sum} Observe that the gap between $O(nw_{\max})$ and $O(Nw_{\max})$ arises from $u_{\max}$ as $N \leq nu_{\max}$. A standard method for reducing $u_{\max}$ is to bundle item copies. For example, by bundling every $k$ copies of each item, we reduce $u_{\max}$ by a factor of $k$.  Doing this, however, increases $w_{\max}$ by the same factor.  As a consequence, any trivial treatment of the resulting instance will not improve the running time. Let $X$ be the input instance. Let ${\cal S}(X)$ be the set of all subset sums of $X$. The effect of bundling is essentially decomposing $X$ into $X = X_0 + kX_1$ where $X_0$ stands for the set of the items that are not bundled (called residual items), and $kX_1$ stands for the set of the bundled items. We can prove that ${\cal S}(X) = {\cal S}(X_0) + k {\cal S}(X_1)$. Since both the residual items and the bundled items have smaller multiplicities, computing ${\cal S}(X_0)$ and ${\cal S}(X_1)$ is faster than directly computing ${\cal S}(X)$. Computing ${\cal S}(X_0) + k{\cal S}(X_1)$ using standard FFT, however, still requires $\widetilde{O}(\texttt{m}_{{\cal S}(X_0)} + k\cdot \texttt{m}_{{\cal S}(X_1)})$ time where $\texttt{m}_{{\cal S}(X_i)}$ is the maximum element in ${\cal S}(X_i)$. Note that $\texttt{m}_{{\cal S}(X_0)} + k\cdot \texttt{m}_{{\cal S}(X_1)}$ can be as large as $Nw_{\max}$, so it leads to the failure of the bundling idea.
We propose an FFT-based algorithm that can determine whether $t\in {\cal S}(X_0) + k{\cal S}(X_1)$ in $\widetilde{O}(\texttt{m}_{{\cal S}(X_0)} +  \texttt{m}_{{\cal S}(X_1)})$ time for arbitrary $t$.  This algorithm can be extended to arbitrary many levels of bundling. After recursively applying the bundling idea for logarithmic levels, we can reduce the multiplicities of items in each level to a constant. Therefore, the ${\cal S}(X_i)$ for each level can be computed in $\widetilde{O}(nw_{\max})$ time. Then we can determine whether  $t \in \sum_i k^i{\cal S}(X_i)$ using our FFT-based algorithm in $\widetilde{O}(nw_{\max})$ total time.

Then we observe that when the number of distinct weights in the input is bounded by $\widetilde{O}(w^{1/2}_{\max})$, after $O(n)$-time preprocessing, Subset sum can be solved in $\widetilde{O}(w^{3/2}_{\max})$ time using our $\widetilde{O}(nw_{\max})$-time algorithm. It remains to tackle the case where the number of distinct weights is at least $\widetilde{\Theta}(w^{1/2}_{\max})$. In this case, the set of item weights is ``dense'', so we can apply additive combinatorics tools as we did for Knapsack.  Moreover, since every item has the same efficiency in Subset Sum, we can get stronger proximity result. Indeed, we show that there is a greedy solution $\veg$ and an optimal solution $\vez$ that differ by at most $\Delta = \widetilde{O}(w^{3/2}_{\max})$ volume of items.  Obtaining $\vez$ from $\veg$ can be basically reduced to two Subset Sum problems with $t = \widetilde{O}(w^{3/2}_{\max})$. Then using Bringmann's $\widetilde{O}(n+t)$-time algorithm \cite{Bri17} for Subset Sum, the dense case can be solved in $\widetilde{O}(n + w^{3/2}_{\max})$ time.

\subsection{Further Related Work}
For Knapsack, $N, t, p_{\max}, OPT$ are alternative parameters that have been used in the literature. The standard dynamic programming due to Bellman together with the idea of bundling item copies gives a running time of $\widetilde{O}(nt)$~\cite{Bel57,Pis99}. Using a balancing technique, Pisinger obtained a dynamic programming algorithm that runs in $O(Nw_{\max}p_{\max})$ time~\cite{Pis99}.  Several recent advances in Knapsack build upon $(\min, +)$-convolution~\cite{Cha18, AT19,BHSS18,BC22,KP04}. In particular, Bringmann and Cassis used the partition and convolve paradigm to develop an algorithm that runs in  $\widetilde{O}(N w_{\max} p_{\max}^{2/3}\})$ time~\cite{BC23}. The additive combinatorics techniques used in these papers also inspired progress in approximation algorithms of Knapsack~\cite{DJM23}.  Pseudopolynomial-time algorithms have also been studied extensively for Knapsack and Subset Sum in the unbounded case where $u_i=\infty$. See Reference~\cite{BHSS18, Bri17, BC22, CH22, JR19, Tam09}.

For Subset Sum, when taking the target $t$ as a parameter, the best-known deterministic algorithm is due to Koiliaris and Xu~\cite{KX19}, which runs in $\widetilde{O}(n+\min\{\sqrt{n}t, t^{5/4}\})$ time. The best-known randomized algorithm is due to Bringmann~\cite{Bri17}, which runs in $\widetilde{O}(n+t)$ time. An alternative randomized algorithm with almost the same running time is given by Jin and Wu~\cite{JW18}. 

\subsection{Paper Outline}
In Section~\ref{sec:Preliminaries}, we introduce necessary terminology and preliminaries.  In Section~\ref{sec:exchange}, we prove two exchange arguments that are crucial to our algorithms. In Section~\ref{sec:Knapsack-2.5}, we present a simpler algorithm to illustrate our main idea. The $\widetilde{O}(n + w^{12/5}_{\max})$-time algorithm for Knapsack is presented in Section~\ref{sec:Knapsack-2.4}. In Section~\ref{sec:Subset Sum}, we give two algorithms for Subset Sum. Omitted proofs can be found in the appendix.

\section{Preliminaries}\label{sec:Preliminaries}

\subsection{Notation} 

We use ${\cal I} = \{1, \ldots, n\}$ to denote the set of all items, each with weight $w_i$ and profit $p_i$. We assume that items are labeled in decreasing order of efficiency. That is, $p_1 / w_1 \geq \ldots \geq p_n/ w_n$. Item $i$ has $u_i$ copies.  Let $W= \{w_i : i \in {\cal I}\}$ be the set of all item weights. The maximum weight is denoted by $w_{\max}$.

Let $i$ be some item. We denote the set of item $\{1, ... , i-1\}$ as ${\cal I}_{< i}$ and $\{i, \ldots, n\}$ as ${\cal I}_{\geq i}$. Let $i$ and $j$ be two items with $i < j$, we say $i$ is to the left of $j$, and $j$ is to the right.

For $w \in W$, we denote the set of items with weight $w$ as 
\[
     {\cal I}^{w} = \{i \in {\cal I}: w_i = w\}.   
\]
For any subset $W'$ of $W$, we write $\bigcup_{w\in W'} {\cal I}^w$ as ${\cal I}^{W'}$. We also write ${\cal I}^{W'} \cap {\cal I}_{< i}$ as ${\cal I}^{W'}_{< i}$, and ${\cal I}^{W'} \cap {\cal I}_{\geq  i}$ as ${
\cal I}^{W'}_{\geq i}$. Basically, the superscript restricts the weights of items in the set, while the subscript restricts the indices.

For any subset ${\cal I}'$ of items, we refer to the multiset $\{w_i : i \in {\cal I}'\}$ as the weight multiset of ${\cal I}'$.

A solution to Bounded Knapsack can be represented by a vector $\vex$ where $x_i$ is the number of selected copies of item $i$.

\subsection{The Previous Proximity Result} 

Let $\veg$ be a maximal prefix solution that can be obtained greedily as follows. Starting with an empty knapsack, we process items in increasing order of index (i.e.,~in decreasing order of efficiency). If the knapsack has enough room for all copies of the current item, we select all these copies, and process the next item. Otherwise, we select as many copies as possible, and then stop immediately. The item at which we stop is called the \emph{break item}, denoted by $b$. Without loss of generality, we assume that no copies of the break item $b$ are selected by ${\ve g}$ (i.e., $g_b = 0$). If $g_{b} \neq 0$, we can view those copies selected by ${\ve g}$ as a new item $i'$ with $u_{i'} = g_{b}$, and set $u_{b} = u_{b} - g_{b}$. It is easy to observe that $g_i = u_i$ for all $i < b$ and that $g_i = 0$ for all $i \geq b$.

The following proximity result states that there is an optimal solution ${\ve z}$ that differs from ${\ve g}$ only in a few items. 

\begin{lemma}[\cite{EW19, PRW21}]\label{lem:proximity}
  Let $\veg$ be a maximal prefix solution of Knapsack. There exists an optimal solution $\vez$ such that 
  \[
          \|\vez - \veg\|_1 = \sum_{i=1}^{n} |z_i - g_i| \leq 2w_{\max}.
    \]
\end{lemma}

Throughout the rest of the paper, we fix $\veg$ to be a maximal prefix solution and $\vez$ to be an optimal solution that minimizes $\|\vez - \veg\|_1$. For any subset ${\cal I}'$ of items, we define the weighted $\ell_1$-distance from $\veg$ to $\vez$ on ${\cal I}'$ to be
\[
    \Delta({\cal I}') = \sum_{i \in {\cal I}'} w_i|g_i - z_i|.
\]

\subsection{Proximity-based $(\min, +)$-Convolution Framework}\label{subsec:conv}

 

Both the algorithm of Polak et al.~\cite{PRW21} and our algorithm use the proximity-based $(\min, +)$-convolution framework.
We present a general form of the framework through the following Lemma.  Recall that $W$ is the set of all item weights.
\begin{restatable}{lemma}{lemconvbygroup}
    \label{lem:conv-by-group}
    Let $\veg$ be a maximal prefix solution to Bouned Knapsack. Let $W_1 \cup \cdots \cup W_k$ be a partition of the set $W$. For $j \in \{1,...,k\}$, let $U_j$ be an upper bound for $\Delta({\cal I}^{W_j})$. 
    Then Bounded Knapsack can be solved in $ \widetilde{O}(n + k\sum_{j = 1}^k |W_j|\cdot U_j)$ time.
\end{restatable}

We only give a sketch of the proof. The complete proof is deferred to Appendix~\ref{Appendix proof lem:conv-by-group}.
\begin{Proof Sketch}
 We construct an optimal solution $\vex$ through modifying $\veg$. Specifically, we have  $\vex = \veg-\vex^-+\vex^+$, where $\vex^-$ represents the copies of items in ${\cal I}_{< b}$ that are deleted from $\veg$, and $\vex^+$ represents the copies of items in ${\cal I}_{\geq b}$ that are added to $\veg$, while noting that $b$ is the break item. We shall compute two sequences, namely: (i). $\boldsymbol{x}^{-} = \langle x^{-}_0,...,x^{-}_{k\cdot U_k}\rangle$, where $x^{-}_{t'}$ is the minimum total profit of copies in $\mathcal{I}_{< b}$ whose total weight is exactly $t'$; (ii). $\boldsymbol{x}^{+} = \langle x^{+}_{ 0},...,x^{+}_{ k\cdot U_k} \rangle$, where $x^{+}_{t'}$ is the maximum total profit of copies in $\mathcal{I}_{\geq b}$ whose total weight is exactly $t'$. The optimal solution can be found easily using these two sequences in $O(kU_k)$ time.  

 The computation of $\boldsymbol{x}^{-}$ and $\boldsymbol{x}^{+}$ follows the same method as Polak et al.~\cite{PRW21}: we consider all copies of items whose weight is exactly $w$, and let $\boldsymbol{s}^w$ denote the sequence whose $t'$-th entry is the maximum total profit of these items when their total weight is exactly $t'$. After computing all $\boldsymbol{s}^w$'s, we iteratively update $\boldsymbol{x}^{-}$ and $\boldsymbol{x}^{+}$ by computing their convolutions with each $\boldsymbol{s}^w$. The key point is that we can strategically pick an order of the  $\boldsymbol{s}^w$'s to perform convolution. In particular,  
assuming that $U_1\leq ...\leq U_k$, we have $\Delta(\mathcal{I}^{W_1\cup ...\cup W_{k'}})\leq \sum_{j=1}^{k'} U_j \leq k' \cdot U_{k'}$ for any $1\leq k'\leq k$. We first convolve $\boldsymbol{s}^w$'s in $W_1$, then $W_2$, etc. When we compute the convolution with $\boldsymbol{s}^w$ where $w\in W_{k'}$, we can truncate the sequence after the $k'\cdot U_{k'}$-th entry, and thus the convolution takes $O(k'\cdot U_{k'})$ time via SMAWK algorithm \cite{AKM+87}.
It is easy to verify that the total time is $\widetilde{O}(n+\sum_{j=1}^{k}|W_j|\cdot j\cdot U_j)=\widetilde{O}(n+k\sum_{j=1}^{k}|W_j|\cdot  U_j) $.
\end{Proof Sketch}

\subsection{Additive Combinatorics}
Additive combinatorics has been a useful tool for Subset Sum. With additive combinatorics tools, Galil and Margalit~\cite{GM91} gave a characterization of the regime within which Subset Sum can be solved in linear time. Their result was later improved by Bringmann and Wellnitz~\cite{BW21}. Our exchange argument in Section~\ref{sec:exchange} heavily utilizes the dense properties by Bringmann and Wellnitz.  We briefly descibe these properties in this subsection.

Let $X$ be a multiset of positive integers. We denote the sum of $X$ as $\Sigma_X$, the maximum element in $X$ as $\texttt{m}_X$, the maximum multiplicity of elements in $X$ as $\mu_X$, and the support of $X$ as $\texttt{supp}_X$. The number of elements in $X$ is denoted by $|X|$ (with multiplicities counted). We say $X$ can \emph{hit} an integer $s$ if some subset of $X$ sums to this integer, i.e., there is a subset $X'\subseteq X$ that $\Sigma_{X'} = s$. 

\begin{Definition}[\normalfont Definition 3.1 in \cite{BW21}]
\label{def:dense}
We say that a multiset $X$ is $\delta$-dense if it satisfies {\normalfont $|X|^2 \geq \delta \cdot \mu_X\texttt{m}_X$}.
\end{Definition}

\begin{Definition}[\normalfont Definition 3.2 in \cite{BW21}]
\label{def:divisor}
Let $X$ be a multiset. We denote by $X(d):= X \cap d\mathbb{Z}$ the multiset of all numbers in $X$ that are divisible by $d$. Further, we write $\overline{X(d)}:= X \backslash X(d)$ to denote the multiset of all numbers in $X$ not divisible by $d$. We say an integer $d > 1$ is an $\alpha$-almost divisor of $X$ if $|\overline{X(d)}|\leq\alpha\ \cdot\mu_X\Sigma_X/|X|^2$.
\end{Definition}

\begin{theorem}[\normalfont Theorem 4.2 in \normalfont \cite{BW21}]
\label{thm:hitrange}
Let $X$ be a multiset of positive integers and set
\begin{align*}
c_\delta &:= 1699200\cdot \log(2|X|)\log^2(2\mu_X), \\
c_\alpha &:= 42480\cdot\log(2\mu_X),\\
c_\lambda &:= 169920\cdot\log(2\mu_X).  
\end{align*}
If X is $c_\delta$-dense and has no $c_\alpha$-almost divisor, then for \normalfont{$\lambda_X := c_\lambda \mu_X\texttt{m}_X{\Sigma}_{X}/|X|^2$}, $X$ can hit all the integers in range
\(
[ \lambda_X, \Sigma_X -\lambda_X ].
\)
\end{theorem}

\begin{theorem}[\normalfont Theorem 4.1 in \normalfont \cite{BW21}]
\label{thm:divisor}
Let $X$ be a multiset of positive integers. 
Let $\delta,\alpha$ be functions of $n$ with $\delta\ge1$ and $0< \alpha\leq\delta/16$. Given a $\delta$-dense set $X$ of size $n$, there exists an integer $d\geq 1$ such that $X':= X(d)/d$ is $\delta$-dense and has no $\alpha$-almost divisor. Moreover, we have the following additional properties:
\begin{enumerate}[label = {\normalfont (\roman*)}]
    \item $d \leq 4\mu_X\Sigma_X/|X|^2$, 
    \item $|X'| \geq 0.75 |X|$,
    \item $\Sigma_{X'}\ge0.75\Sigma_X/d$.
\end{enumerate}
\end{theorem}

\section{A General Exchange Argument}\label{sec:exchange}
In this section, we establish two additive combinatorics results that are crucial to proving our proximity result. Basically we show that given two (multi-)sets $A$ and $B$ of positive integers, if $A$ is dense, and $\Sigma_B$ is large, there must be some non-empty subset $A'$ of $A$ and a subset $B'$ of $B$ such that $\Sigma_{A'} = \Sigma_{B'}$. We first give a result where multiplicities of integers are not utilized.

\begin{lemma}
    \label{lem:hit-a}
    Let $w$ be a positive integer. Let $p$ be an arbitrary real number such that $0.5\leq p < 1$. Let $A$ be a multiset of integers from $\{1,..., w\}$ such that $|\texttt{supp}_A| \geq c_A w^{p} \log w$, $B$ be a multiset of integers from $\{1,..., w\}$ such that $\Sigma_B \geq c_B w^{2-p}$, where $c_A$ and $c_B$ are two sufficiently large constants. Then there must exist a non-empty subset $A'$ of $A$ and a non-empty subset $B'$ of $B$ such that $\Sigma_{A'} = \Sigma_{B'}$. 
\end{lemma}
\begin{proof}
Let $S = \texttt{supp}_A$. Note that $|S| \geq c_A w^{p} \log w $. What we actually prove is that there is a non-empty subset $S'$ of $S$ such that some subset $B'$ of $B$ have the same total weight as $S'$.
We first characterize the set of integers that can be hit by $S$. Then we show that $B$ can hit at least one of these integers.

    Note that $\texttt{m}_S \leq w$, $\mu_S = 1$, and $|S| \leq w$.  Since $c_A$ is sufficiently large, we have that
    \[
        |S|^2 \geq c^2_A w^{2p} \log^2 w \geq c^2_A w \log w \geq c^2_A \mu_S 
        \texttt{m}_S \log |S| \geq c_\delta \mu_S \texttt{m}_S.
    \]
    By definition, $S$ is $c_\delta$-dense. By Theorem~\ref{thm:divisor}, there exists an integer $d$ such that  $S':= S(d)/d$ is $c_{\delta}$-dense and has no $c_\alpha$-almost divisor. And the followings hold.
    \begin{enumerate}[label = {\normalfont (\roman*)}]
        \item $d \leq 4 \mu_S \Sigma_S/|S|^2$, 
        \item $|S'| \geq 0.75 |S|$.
        \item $\Sigma_{S'}\geq 0.75\Sigma_S/d$.
    \end{enumerate}
    Note that $\mu_{S'} = 1$, $\texttt{m}_{S'} \leq w/d$, and $\Sigma_{S'} \leq \Sigma_S / d$. Applying Theorem~\ref{thm:hitrange} on $S'$, we get $S'$ can hit any integer in  the range $[\lambda_{S'}, \Sigma_{S'} - \lambda_{S'}]$ where 
    \[
        \lambda_{S'} = \frac{c_{\lambda} \mu_{S'}\texttt{m}_{S'}\Sigma_{S'}}{|S'|^2} \leq \frac{c_{\lambda} w \Sigma_S}{(0.75|S|)^2 d^2} \leq \frac{\min\{c_A, c_B\}}{2} \cdot \frac{w \Sigma_S}{d^2|S|^2}.
    \]
    The last inequality holds since $c_A$ and $c_B$ are sufficiently large constants. We can conclude that $S$ can hit any multiple of $d$ in the range $[d\lambda_{S'}, d(\Sigma_{S'} - \lambda_{S'})]$. We also have that the left endpoint of this interval
    \[
        d\lambda_{S'} \leq \frac{c_B}{2} \cdot \frac{w\Sigma_S}{d |S|^2} \leq \frac{c_B}{2} \cdot \frac{w^2 |S|}{|S|^2} \leq \frac{c_B}{2} \cdot  \frac{w^2}{|S|}\leq \frac{c_B}{2} \cdot w^{2-p},
    \]
    and that the length of the interval 
\begin{align*} 
    d (\Sigma_{S'} - 2\lambda_{S'}) 
    \geq &\frac{3\Sigma_S}{4} - c_A\cdot \frac{w \Sigma_S}{d|S|^2}\\
    \geq &\frac{\Sigma_S}{|S|^2}(\frac{3|S|^2}{4} - c_A\cdot \frac{w}{d})  & &(\text{since } |S| \geq c_Aw^{1/2}\log w\text{ and } d \geq 1)\\
    \geq &\frac{\Sigma_S}{|S|^2}(\frac{3c_A^2}{4}\cdot w - c_A\cdot w)  & &( \text{since $c_A$ is sufficiently large})\\
    \geq &4\cdot \frac{\Sigma_S}{|S|^2}\cdot w  & &(\text{since $ d \leq 4\mu_S\Sigma_S/|S|^2$ and $\mu_S = 1$})\\
    \geq &dw
\end{align*}

To complete the proof, it suffices to show that there is a subset $B'$ of $B$ whose sum is a multiple of $d$ and is within the interval $[d\lambda_{S'}, d(\Sigma_{S'} - \lambda_{S'})]$. We claim that as long as $B$ has at least $d$ numbers, there must be a non-empty subset of $B$ whose sum is at most $dw$ and is a multiple of $d$. Assume the claim is true. We can repeatedly extract such subsets from $B$ until $B$ has less than $d$ numbers. Note that the total sum of these subsets is at least
\[
  \Sigma_B - wd \geq c_B w^{2-p} - w\cdot \frac{4\Sigma_S}{|S|^2}\geq c_B w^{2-p} -  \frac{4w^2}{|S|} \geq c_B w^{2-p} - w^{2-p} \geq \frac{c_Bw^{2-p}}{2}.
\]
That is, the total sum of these subsets is at least the left endpoint of $[d\lambda_{S'}, d(\Sigma_{S'} - \lambda_{S'})]$. Also note that the sum of each subset is at most $dw$, which does not exceed the length of the interval. As a result, there must be a collection of subsets whose total sum is within the interval.  Since the sum of each subset is a multiple of $d$, so is any collection of these subsets.
 
To see why the claim is true, take $d$ arbitrary numbers from $B$. For $i \in \{1,..., d\}$, Let $h_i$ be the sum of the first $i$ numbers. If any of $h_i \equiv 0 \pmod d$, then we are done. Otherwise, by the pigeonhole principle, there must be $i < j$ such that $h_i \equiv h_j \pmod d$. This implies that there is a subset of $j - i$ numbers whose sum is $h_j - h_i \equiv 0 \pmod d$. Note that $0 < j - i \leq d$. So this subset is non-empty and has its sum at most $dw$.
\end{proof}

Lemma~\ref{lem:hit-a} has no requirement on the multiplicities of elements in $A$. Therefore, in order to be dense, $A$ must contain at least $c_A w^{1/2} \log w$ distinct integers. The following Lemma~\ref{lem:hit-b} generalizes Lemma~\ref{lem:hit-a} by taking the multiplicity of integers into consideration. Its proof is similar to that of Lemma~\ref{lem:hit-a}, so we defer it to the appendix~\ref{Appendix proof lem:hit-b}.

\begin{restatable}{lemma}{lemhitb}
  \label{lem:hit-b}
  Let $w$ be a positive integer. Let $A$ and $B$ be two multisets of integers from $\{1,...,w\}$ such that
  \begin{enumerate}[label = {\normalfont (\roman*)}]
    \item at least $c_A w^{2/5} \log^2 w$ distinct integers in $A$ have multiplicity of at least $w^{1/5}$,

    \item $\Sigma_B \geq c_Bw^{8/5}$,
  \end{enumerate}
 where $c_A$ and $c_B$ are two sufficiently large constants.   Then there must exist a non-empty subset $A'$ of $A$ and a non-empty subset $B'$ of $B$ such that $\Sigma_{A'} = \Sigma_{B'}$. 
\end{restatable}


\section{An $\widetilde{O}(n + w^{5/2}_{\max})$-time Algorithm for Knapsack}\label{sec:Knapsack-2.5}
To illustrate our main idea, we first present a simpler algorithm that runs in $\widetilde{O}(n + w^{5/2}_{\max})$.  The algorithm uses the proximity-based $(\min,+)$-convolution framework (See Subsection~\ref{subsec:conv}). By Lemma~\ref{lem:conv-by-group}, the key to obtaining a fast algorithm is to find a good partition of $W$. 

Polak~et al.~\cite{PRW21} did not partition $W$ at all. That is, they used $k = 1$ and $W_1 = W$. Lemma~\ref{lem:proximity} implies an upper bound of $2w^2_{\max}$ on $\Delta({\cal I}^{W})$. Together with the fact that $|W| \leq w_{\max}$, they obtained a running time of $O(n + w^3_{\max})$.  Although the upper bound of $O(w^2_{\max})$ on $\Delta({\cal I}^{W})$ is actually tight, an improvement in the running time is still possible. Indeed, our first proximity result states that $\Delta({\cal I}^{W})$ concentrates at the small set $W^*$ of roughly $w^{1/2}_{\max}$ weights, and that the rest of the weights together contribute at most $O(w^{3/2}_{\max})$. Moreover, $W^*$ can be identified in $\widetilde{O}(n)$ time. Then an $\widetilde{O}(n + w^{5/2}_{\max})$-time algorithm directly follows by Lemma~\ref{lem:conv-by-group}.

This section is divided into two parts: the structural part and the algorithmic part. The structural part proves the existence of $W^*$, and the algorithmic part gives algorithms for finding $W^*$ and solving Knapsack. 

\subsection{Structural Part -- Fine-Grained Proximity}\label{subsec:2.4-struct}
In the structural part, Bounded Knapsack and 0-1 Knapsack are equivalent, in the sense that every item with $u_i$ copies can be viewed as $u_i$ items.  Therefore, only 0-1 Knapsack is discussed in the structural part. In 0-1 Knapsack, each item $i$ has only one copy, so there is no difference between items and copies of items. We can use them interchangeably. As a result, any solution to 0-1 Knapsack can be represented as a subset of ${\cal I}$. Let $\vex \in \{0,1\}^{|\mathcal{I}|}$ be a solution to 0-1 Knapsack. With slight abuse of notation, we also use $\vex$ to denote the set of items selected by $\vex$.  For any subset ${\cal I}'$ of items, $\vex({\cal I}') = {\cal I}' \cap \vex$ and $\overline{\vex}({\cal I}') = {\cal I}' \setminus \vex$ denote the sets of items from ${\cal I}'$ that are selected and not selected by $\vex$, respectively.

Our first proximity result is the following.
\begin{lemma}\label{lem:2.5-parition-struct}
    There exists a partition $(W^*, \overline{W^*})$ of $W$ such that 
    \begin{enumerate}[label = {\normalfont (\roman*)}]
        \item $|W^*| = 4c_A w^{1/2}_{\max}\log w_{\max}$, and

        \item $\Delta({\cal I}^{\overline{W^*}}) \leq 4c_Bw_{\max}^{3/2}$,
    \end{enumerate}
    where $c_A$ and $c_B$ are two large constants used in Lemma~\ref{lem:hit-a}.
\end{lemma}
\begin{proof}
    We will define $W^*$ via a partition of ${\cal I}$, and then show that $W^*$ satisfies properties in the lemma.

    \paragraph{Defining $W^*$ via a partition of ${\cal I}$.} Recall that we label the items in decreasing order of efficiency. That is, $p_1 / w_1 \geq \ldots \geq p_n/ w_n$.  Without loss of generality, we assume that for any two items $i < j < b$, if $i$ and $j$ have the same efficiency, then $w_i \leq w_j$, and that for any two items $b < i < j$, if $i$ and $j$ have the same efficiency, then $w_i \geq w_j$. We partition the items into four groups $({\cal I}_1, {\cal I}_2, {\cal I}_3, {\cal I}_4)$ according to their indices. See Figure~\ref{fig:partition} for an illustration.  We shall define ${\cal I}_2$ and ${\cal I}_3$ in such a way that (i) they are ``close'' to the break item $b$ and that (ii) items within them have $\Theta(w_{\max}^{1/2}\log w_{\max})$ distinct weights separately. Because the items in ${\cal I}_2$ and ${\cal I}_3$ have their efficiency close to that of $b$, it is difficult to tell how many of them should be selected by the optimal solution. In other words, $\vez$ and $\veg$ may differ a lot in these items. In contrast, the items in ${\cal I}_1$ and ${\cal I}_4$ have very high and very low efficiency, respectively, so $\vez$ tends to select most of ${\cal I}_1$ and few of ${\cal I}_4$ . As a result, $\vez$ and $\veg$ are similar in ${\cal I}_1$ and ${\cal I}_4$.
    
    \begin{figure}[ht]
        \centering
        \begin{tikzpicture}
        \draw [thick] (0,0) -- (10,0);
        \draw [thick] (0, 0.1) -- (0, -0.1);
        \draw [thick] (3, 0.1) -- (3, -0.1);
        \draw [thick] (5, 0.1) -- (5, -0.1);
        \draw [thick] (7, 0.1) -- (7, -0.1);
        \draw [thick] (10, 0.1) -- (10, -0.1);  

        \node [align =center, below] at (5, -0.2) {$b$};
        \node [align =center, below] at (0, -0.2) {$0$};
        \node [align =center, below] at (10, -0.2) {$n$};
        \node [align = center, above] at (1.5, 0) {${\cal I}_{1}$};
        \node [align = center, above] at (4, 0) {${\cal I}_{2}$};
        \node [align = center, above] at (6, 0) {${\cal I}_{3}$};
        \node [align = center, above] at (8.5, 0) {${\cal I}_{4}$};   

        \draw [thick, ->] (2, - 1) -- (8, - 1);
        \node [align = center, below] at (5,-1.2) {larger index and lower efficiency}; 
        \end{tikzpicture}
        \caption{${\cal I}_{1}$, ${\cal I}_{2}$, ${\cal I}_{3}$, and ${\cal I}_{4}$}
        \label{fig:partition}
    \end{figure}

    Now we formally describe $({\cal I}_1, {\cal I}_2, {\cal I}_3, {\cal I}_4)$.  For any $i < j$, let ${\cal I}_{[i,j)}$ be the set of items $\{i, \ldots, j-1\}$. Let $i^*$ be the minimum index $i$ such that the items in ${\cal I}_{[i,b)}$ have exactly $2c_Aw_{\max}^{1/2}\log w_{\max}$ distinct weights. Let ${\cal I}_2 = {\cal I}_{[i^*, b)}$, and let ${\cal I}_1 = {\cal I}_{< i^*}$. When no such $i^*$ exists, let ${\cal I}_2 = {\cal I}_{<b}$, and let ${\cal I}_1 = \emptyset$. ${\cal I}_3$ and ${\cal I}_4$ are defined similarly as follows. Let $j^*$ be the maximum index $j$ such that the items in ${\cal I}_{[b,j)}$ have exactly $2c_Aw_{\max}^{1/2}\log w_{\max}$ distinct weights. Let ${\cal I}_3 = {\cal I}_{[b,j^*)}$, and let ${\cal I}_4 = {\cal I}_{\geq j^*}$. When no such $j^*$ exists, let ${\cal I}_3 = {\cal I}_{\geq b}$, and let ${\cal I}_4 = \emptyset$. 

    $W^*$ is defined as the set of the weights of the items in ${\cal I}_2 \cup {\cal I}_3$. 

    \paragraph{Verifying properties.} It is straightforward that $|W^*| \leq 4c_A w^{1/2}_{\max}\log w_{\max}$. We are left to show $\Delta({\cal I}^{\overline{W^*}}) \leq 2c_B w_{\max}^{3/2}$.  Note that ${\cal I}^{\overline{W^*}} \subseteq {\cal I}_1 \cup {\cal I}_4$, so ${\cal I}^{\overline{W^*}}$ can be partitioned into ${\cal I}^{\overline{W^*}}_1 = {\cal I}^{\overline{W^*}} \cap {\cal I}_1$ and ${\cal I}^{\overline{W^*}}_4 = {\cal I}^{\overline{W^*}} \cap {\cal I}_4$, we have $\Delta({\cal I}^{\overline{W^*}}) = \Delta({\cal I}^{\overline{W^*}}_1) + \Delta({\cal I}^{\overline{W^*}}_4)$.  It suffices to show that $\Delta({\cal I}^{\overline{W^*}}_1)$ and $\Delta({\cal I}^{\overline{W^*}}_4)$ are both bounded by $2c_B w^{3/2}_{\max}$. In the following, we only provide proof for $\Delta({\cal I}^{\overline{W^*}}_1)$. Bounds for $\Delta({\cal I}^{\overline{W^*}}_4)$ can be proved similarly due to symmetry.
    
    Suppose, for the sake of contradiction, that $\Delta({\cal I}_1^{\overline{W^*}}) > 2c_B w^{3/2}_{\max}$. That is, a great volume of items in ${\cal I}_1^{\overline{W^*}}$ are deleted when obtaining the optimal solution $\vez$ from the greedy solution $\veg$. (Note that $\veg$ selects all items in ${\cal I}_1^{\overline{W^*}}$ as ${\cal I}^{\overline{W^*}}_1 \subseteq {\cal I}_{< b}$, so $\Delta({\cal I}^{\overline{W^*}}_1)$ is exactly the total weight of items in ${\cal I}_1$ that are deleted.)  Obviously, ${\cal I}^{\overline{W^*}}_1$ is non-empty since otherwise $\Delta({\cal I}^{\overline{W^*}}_1)$ would be $0$. This implies that ${\cal I}_2$ has exactly $2c_A w_{\max}^{1/2} \log w_{\max}$ distinct weights. The greedy solution $\veg$ picks all the items in ${\cal I}_2$ as ${\cal I}_2 \subseteq {\cal I}_{< b}$. When obtaining $\vez$ from $\veg$, at least one of the following two cases must be true.
    \begin{enumerate}[label={(\roman*)}]
        \item $\vez$ keeps lots of items in ${\cal I}_2$ so that $\vez({\cal I}_2)$ have at least $c_A w_{\max}^{1/2} \log w_{\max}$ distinct weights.
  
        \item lots of the items in ${\cal I}_2$ are deleted so that $\vez({\cal I}_2)$ have at most $c_A w_{\max}^{1/2} \log w_{\max}$ distinct weights. In other words, $\overline{\vez}({\cal I}_2)$ have at least $c_A w_{\max}^{1/2} \log w_{\max}$ distinct weights.
    \end{enumerate}

    We will show that both cases lead to a contradiction. The high-level idea is the following.
    In case (i), when obtaining $\vez$, lots of items in ${\cal I}_2$ are kept while a great volume of items in ${\cal I}^{\overline{W^*}}_1$ are deleted. We will show that there is a subset ${\cal Z} \subseteq \vez({\cal I}_2)$ of items kept by $\vez$ and a subset ${\cal D} \subseteq \overline{\vez}({\cal I}^{\overline{W^*}}_1)$ of the deleted items such that ${\cal D}$ and ${\cal Z}$ have the same total weight. Note that items in ${\cal D}$ have higher efficiency than those in ${\cal Z}$ as ${\cal D}$ is to the left of ${\cal Z}$. But then, when obtaining $\vez$ from $\veg$, why not delete ${\cal Z}$ rather than ${\cal D}$? In case (ii), lots of items in ${\cal I}_2$ are deleted. Also, since a great volume of items (in ${\cal I}^{\overline{W^*}}_1$) are deleted, at least the same volume of items should be added in order to obtain the optimal solution $\vez$.  We will show that there is a subset ${\cal D} \subseteq \overline{\vez}({\cal I}_2)$ of the deleted items and a subset ${\cal A}\subseteq \vez({\cal I}_{\geq b})$ of the added items such that ${\cal A}$ and ${\cal D}$ have the same total weight. Since the items in ${\cal D}$ have higher efficiency than those in ${\cal A}$, when obtaining $\vez$ from $\veg$, why bother to delete ${\cal D}$ and add ${\cal A}$?

    \paragraph{Case (i).} $\vez({\cal I}_2)$ has at least $c_A w_{\max}^{1/2} \log w_{\max}$ distinct weights. Note that $\overline{\vez}({\cal I}^{\overline{W^*}}_1)$ is exactly the set of items in ${\cal I}^{\overline{W^*}}_1$ that are deleted when obtaining $\vez$ and the total weight of $\overline{\vez}({\cal I}^{\overline{W^*}}_1)$ is $\Delta({\cal I}^{\overline{W^*}}_1) > c_B w^{3/2}_{\max}$. The weight multisets of $\vez({\cal I}_2)$ and $\overline{\vez}({\cal I}^{\overline{W^*}}_1)$ satisfy the condition of Lemma~\ref{lem:hit-a}, which implies that there is a non-empty subset ${\cal Z}$ of $\vez({\cal I}_2)$ and a non-empty subset ${\cal D}$ of $\overline{\vez}({\cal I}^{\overline{W^*}}_1)$ such that ${\cal Z}$ and ${\cal D}$ have the same total weight. Let $\hat{\vez} = (\vez \setminus {\cal Z}) \cup {\cal D}$ be the solution obtained from $\vez$ by deleting ${\cal Z}$ and adding ${\cal D}$ back. Clearly $\hat{\vez}$ is feasible. Note ${\cal D}$ is to the left of ${\cal Z}$ (see Figure~\ref{fig:partition} for an illustration). According to the way we label items and break ties, we have that ${\cal D}$ has either strictly higher average efficiency than ${\cal Z}$, or --- if the efficiency is the same --- strictly smaller average weight than ${\cal Z}$. In the former case, $\hat{\vez}$ has a larger profit than $\vez$. In the latter case, $|{\cal D}| > |{\cal Z}|$. Since $\veg$ selects all items in ${\cal D}$ and ${\cal Z}$, it follows that $\hat{\vez}$ is an optimal solution with $\|\hat{\vez} - \veg\|_1 < \|\vez-\veg\|_1$. Both cases contradict our choice of $\vez$.

    \paragraph{Case (ii).} $\overline{\vez}({\cal I}_2)$ has at least $c_Aw_{\max}^{1/2} \log w_{\max}$ distinct weights. When obtaining $\vez$ from $\veg$, the total weight of items that are added is exactly $\Delta({\cal I}_{\geq b})$, while the total weight of items that are deleted is at least $\Delta({\cal I}^{\overline{W^*}}_{1})$. Since the total weight of $\vez$ and $\veg$ differ by at most $w_{\max}$ (they are both maximal solutions), we have the following.
    \[
        \Delta({\cal I}_{\geq b}) \geq \Delta({\cal I}^{\overline{W^*}}_{1}) - w_{\max} \geq (2c_B-1) w^{3/2}_{\max} \geq c_Bw^{3/2}_{\max}.
    \]
    Also note that $\Delta({\cal I}_{\geq b})$ is exactly the total weight of items in $z({\cal I}_{\geq b})$.  Again, the weight multisets of $\overline{\vez}({\cal I}_2)$ and $\vez({\cal I}_{\geq b})$ satisfies the conditions of Lemma~\ref{lem:hit-a}.  By Lemma~\ref{lem:hit-a}, there is a non-empty subset ${\cal D}$ of $\overline{\vez}({\cal I}_2)$ and a non-empty subset ${\cal Z}$ of $\vez({\cal I}_{\geq b})$ such that ${\cal D}$ and ${\cal Z}$ have the same total weight. Moreover, ${\cal D}$ is to the left of ${\cal Z}$. By an argument similar to that in case (i), we can obtain a solution $\hat{\vez} = (\vez \setminus {\cal Z})\cup {\cal D}$ that is better than $\vez$, which contradicts with our choice of $\vez$.  This completes the proof.
\end{proof}

\subsection{Algorithmic Part}
\begin{lemma}
\label{lem:2.5-parition-alg}
In $\widetilde{O}(n)$ time, we can compute the partition $(W^*, \overline{W^*})$ of $W$ in Lemma~\ref{lem:2.5-parition-struct}.
\end{lemma}
It is easy to see that the process of identifying $W^*$ in the proof of Lemma~\ref{lem:2.5-parition-struct} can be generalized to Bounded Knapsack and can be done in $\widetilde{O}(n)$ time.

The following theorem directly follows by Lemma~\ref{lem:2.5-parition-alg} and Lemma~\ref{lem:conv-by-group}.

\begin{theorem}
    There is an $\widetilde{O}(n + w^{5/2}_{\max})$-time algorithm for Bounded Knapsack.
\end{theorem}

\section{An $\widetilde{O}(n + w^{12/5})$-time Algorithm for Knapsack}\label{sec:Knapsack-2.4}
In the structural part of the $\widetilde{O}(n + w^{5/2}_{\max})$-time algorithm, in order to make the weight multiset of $\vez({\cal I}_2)$ (or $\overline{\vez}({\cal I}_2)$) dense, we require that items in it have $\Theta(w^{1/2}_{\max}\log w_{\max})$ distinct weights. The multiplicities of these weights are not considered. In this section, we develop a better proximity result by exploiting the multiplicities of weights, and therefore obtain an algorithm with improved running time.  As before, this section is divided into the structural part and the algorithmic part.

\subsection{Structural Part - Fine-Grained Proximity}
As in Subsection~\ref{subsec:2.4-struct}, it suffices to consider 0-1 Knapsack, and the notations defined at the beginning of that subsection will be used.

\begin{lemma}\label{lem:2.4-partition-a}
There exists a partition of $W$ into $W^*$ and $\overline{W^*}$ such that 
    \begin{enumerate}[label = {\normalfont (\roman*)}]
        \item $|W^*| \leq 4c_Aw^{3/5}_{\max} \log w_{\max}$,
        \item $\Delta({\cal I}^{\overline{W^*}}) \leq 4c_Bw^{7/5}_{\max}$,     
    \end{enumerate}
where $c_A$ and $c_B$ are two large constants used in Lemma~\ref{lem:hit-a}.
\end{lemma}
Lemma~\ref{lem:2.4-partition-a} does not exploit the multiplicities of weights, and can be proved in exactly the same way as Lemma~\ref{lem:2.5-parition-struct}. The trade-off between $|W^*|$ and $\Delta({\cal I}^{\overline{W^*}})$ essentially results from Lemma~\ref{lem:hit-a}. When $|W^*|$ is $\widetilde{\Theta}(w^p_{\max})$ for some $p \geq 1/2$,  $\Delta({\cal I}^{\overline{W^*}})$ is $O(w^{2-p}_{\max})$.

$\overline{W^*}$ already meets the requirement since
\(
     |\overline{W^*}|\Delta({\cal I}^{\overline{W^*}}) \leq c_Bw^{12/5}_{\max}.   
\)
$W^*$, however, may be problematic since $\Delta({\cal I}^{W^*})$ can be as large as $O(w^2_{\max})$.  We further partition $W^*$ into two sets and prove a similar result as before. But this time, we will take advantage of the multiplicities of weights and use Lemma~\ref{lem:hit-b} instead of Lemma~\ref{lem:hit-a}. 
\begin{restatable}{lemma}{lempartitionb}
    \label{lem:2.4-partition-b}
    There exits a partition of $W^*$ into $W^+$ and $\overline{W^+} = W^* \setminus W^+$ such that \begin{enumerate}[label = {\normalfont (\roman*)}]
        \item $|W^+| \leq 4c_A w^{2/5}_{\max}\log^2 w_{\max}$,
        \item $\Delta({\cal I}^{\overline{W^+}}) \leq 8c_Bw^{9/5}_{\max}$,
    \end{enumerate}
    where $c_A$ and $c_B$ are two large constants used in Lemma~\ref{lem:hit-b}.
\end{restatable}

\subsection{Algorithmic Part}
The construction of $W^*$ and $W^+$ can be easily generalized to Bounded Knapsack and can be done in $\widetilde{O}(n)$ time. So we have: 
\begin{lemma}
\label{lem:2.4-partition-alg}
In $\widetilde{O}(n)$ time, we can compute the partition $(W^+, \overline{W^+}, \overline{W^*})$ of $W$ such that
    \begin{enumerate}[label = {\normalfont (\roman*)}]
        \item $|W^+| \leq 4c_A w^{2/5}_{\max}\log^2 w_{\max}$,
        \item $|\overline{W^+}| \leq 4c_Aw^{3/5}_{\max} \log w_{\max}$ and $\Delta({\cal I}^{\overline{W^+}}) \leq 8c_Bw^{9/5}_{\max}$,
        \item  $\Delta({\cal I}^{\overline{W^*}}) \leq 4c_Bw^{7/5}_{\max}$.
    \end{enumerate}
\end{lemma}
It is easy to verify that
\[
    |W^+|\Delta({\cal I}^{W+}) + |\overline{W^+}|\Delta({\cal I}^{\overline{W^+}}) + |\overline{W^*}|\Delta({\cal I}^{\overline{W^*}}) \leq \widetilde{O}(w^{12/5}_{\max}).
\]
Therefore, Theorem~\ref{thm:knapsack} follows by Lemma~\ref{lem:2.4-partition-alg} and~\ref{lem:conv-by-group}.

\section{Algorithms for Bounded Subset Sum}\label{sec:Subset Sum}

We consider Bounded Subset Sum. We have the following observation by Lemma~\ref{lem:proximity}.

\begin{Observation}\label{obs:subset-sum-u}
    Without loss of generality, we may assume $u_i\leq 4w_{\max}$ for all $1\leq i\leq n$.
\end{Observation}

\subsection{An $\widetilde{O}(nw_{\max})$-time Algorithm}
In this subsection, we present an $\widetilde{O}(nw_{\max})$-time algorithm. It builds upon two ingredients: (i). a faster algorithm for computing the sumset of sets with a special structure, and (ii). a layering technique that transforms the input into a special structure.

\subsubsection{Faster sumset algorithm for sets with a special structure}
We first introduce some notations for integer sets. Given an integer (multi-)set $X$ and a number $a$, we let $X-a:=\{x-a:x\in X\}$. We let $\texttt{m}_{X}$ denote the largest number in $X$. We let $rX:=\{rx:x\in X\}$ for any real number $r$. Given two integer multisets $A,B$, we let $A+B := \{a+b: a\in A\cup\{0\}, b\in B\cup\{0\} \}$. It is easy to verify that $r(A+B)=rA+rB$. 

It is known that the sumset $A+B$ can be computed via Fast Fourier Transform (FFT) in $\widetilde{O}(\texttt{m}_{A}+\texttt{m}_{B})$ time. More generally, the following is true: 

\begin{lemma}\label{lemma:sumset}{\normalfont \cite{BN21}}
	Given sets $S_1, S_2,\cdots,S_{\ell} \subset \mathbb{N}$, we can compute $S_1 + S_2 + \cdots + S_{\ell}$ in $O(\sigma\log\sigma \log\ell)$ time, where $\sigma= \texttt{m}_{S_0}+\texttt{m}_{S_1}+\cdots+\texttt{m}_{S_{\ell}}$.
\end{lemma}

For our purpose, given a target value $t$, we want to decide whether $t\in A+kB$ quickly for some positive integer $k$. Using FFT, it takes $O(\texttt{m}_{A}+k\texttt{m}_{B})$ time, which is too much when $k$ is large. We show that $O(\texttt{m}_{A}+\texttt{m}_{B})$ time is sufficient, and the result can be further generalized by the following lemma.
  
\begin{lemma}\label{lemma:k-sum}
			Given sets 
			$S_0, S_1,\cdots,S_{\ell} \subset \mathbb{N}$, an integer $k\in\mathbb{Z}_{>0}$ and a target number $t$, we can determine whether $t\in S_0+ k S_1+ \cdots+ k^{\ell} S_{\ell}$ in $O(\sigma (\ell+1)\log\sigma)$ time, where $\sigma= \texttt{m}_{S_0}+\texttt{m}_{S_1}+\cdots+\texttt{m}_{S_{\ell}}$.
		\end{lemma}
\begin{proof}
We prove by induction. The lemma is obviously true for $\ell=0$. Suppose it is true for $\ell=h-1$. That is, we can determine whether $t\in S_0+ k S_1+ \cdots+ k^{h-1} S_{h-1}$ in $c_1(\sum_{i=0}^{h-1}\texttt{m}_{S_i})(h-1)\log(\sum_{i=0}^{h-1}\texttt{m}_{S_i})$ time for some sufficiently large constant $c_1$. Now we consider the case $\ell=h$. 

Note that $t\in S_0+ k S_1+ \cdots+ k^{h} S_{h}$ if and only if $t=t_0+t'$ where $t_0\in S_0$ and $t'\in k S_1+ \cdots+ k^{h} S_{h}$. It follows that $t'$ is a multiple of $k$. Therefore, $t\equiv t_0 \pmod k$. Let $r_0\in [0,k-1]$ be the residue of $t$ modulo $k$, $\bar{S}_0\subseteq S_0$ be the subset of integers in $S_0$ whose residue is $r_0$. It is clear that every element in $\bar{S}_0-r_0$ is a multiple of $k$, so $\frac{1}{k}(\bar{S}_0-r_0)$ is an integer set. 
Moreover, it follows that $t\in S_0+ k S_1+ \cdots+ k^{h} S_{h}$ if and only if $t\in \bar{S}_0+ k S_1+ \cdots+ k^{h} S_{h}$, or equivalently, if and only if
\[
    \frac{t-r_0}{k}\in \frac{1}{k}\left((\bar{S}_0-r_0)+ k S_1+ \cdots+ k^{h} S_{h}\right)=\frac{1}{k}(\bar{S}_0-r_0)+  S_1+ kS_2+ \cdots+ k^{h-1} S_{h}.
\]
Now we compute $\hat{S}_1:=\frac{1}{k}(\bar{S}_0-r_0)+  S_1$ via FFT, which takes $c_2\sigma_2\log \sigma_2$ time for some constant $c_2$ and $\sigma_2\leq \texttt{m}_{S_0}/k+\texttt{m}_{S_1}$. It thus remains to determine whether
\[
    \frac{t-r_0}{k}\in \hat{S}_1+ kS_2+ \cdots+ k^{h-1} S_{h}.
\]
By the induction hypothesis, this takes $c_1(h-1)\sigma'\log \sigma'$ time where $\sigma'=\texttt{m}_{\hat{S}_1}+\texttt{m}_{S_2}+\cdots+\texttt{m}_{S_{h}}\leq \texttt{m}_{S_0}/k+\texttt{m}_{S_1}+\cdots+\texttt{m}_{S_{h}}$. Therefore, it is easy to verify that the overall running time is
bounded by $c_2\sigma_2\log\sigma_2+c_1(h-1)\sigma'\log \sigma'\leq c_1h \sigma\log\sigma$ when $c_1\geq c_2$.
\end{proof}

\subsubsection{Item Grouping}
We follow a standard idea to bundle item copies into groups. A solution to the grouped instance will take item copies within a group as a whole. Thus,  the key point is to show that, after grouping, we do not lose any solution, and thus the optimal solution. 


Recall that $u_j$ refers to the copy number of item $j$. 
It is obvious that $u_j$ can be represented as $u_j[0]+u_j[1]\cdot 2^1+\cdots+u_j[l_j]\cdot 2^{l_j}$ where $2\leq u_j[i]<4$ for all $0\leq i\leq {\ell}_j-1$ and $u_j[\ell_j]<4$. This can be easily achieved recursively: let $r$ be the residue of $u_j$ modulo $2$; set $u_j[0]=2+r$, $u_j\leftarrow \frac{u_j-(2+r)}{2}$ and repeat the above procedure. 

As $u_j \leq 4w_{\max}$ for all $1\leq j\leq n$, ${\ell}_j \leq 2+\log w_{\max}$.
For ease of discussion, we let $u_j[i]=0$ for $i>{\ell}_j$ and $\ell=\max_j {\ell}_j$.
Now we define $X_i$ as a multiset that consists of exactly $u_j[i]$ copies of weight $w_j$ for every $j$, and note that by $2^iX_i$ we mean the multiset that consists of exact $u_j[i]$ copies of weight $2^iw_j$, which represents a multiset of groups (where each group contains $2^i$ copies of weight $w_j$ but has to be selected as a whole). 

Let $X$ denote the multiset where every $w_j$ occurs $u_j$ times. Recall that ${\cal S}(X)$ denotes the set of all possible subset-sums of multiset $X$.  
We show the following in the appendix.

\begin{restatable}{lemma}{lemsumunion}
\label{lemma:sum-union}
    \[{\cal S}(X)= \sum_{i=0}^{\ell}2^i{\cal S}(X_i) = 2^0{\cal S}(X_0) + \cdots + 2^{\ell}{\cal S}(X_{\ell}).\] 
    That is, for any $t\in {\cal S}(X)$, there exist $t_i\in {\cal S}(X_i)$ such that $t=\sum_{i=0}^{\ell}2^it_i$.
\end{restatable}
Polak et al.~\cite{PRW21} derived a similar lemma for $\ell = 2$.



\subsubsection{Finalizing the $\widetilde{O}(nw_{\max})$-time algorithm}
Now we are ready to present our $\widetilde{O}(nw_{\max})$-time algorithm. 
By Observation~\ref{obs:subset-sum-u}, we have $u_j \leq 4w_{\max}$. 
Consequently, for $\ell=O(\log w_{\max})$, the given instance can be grouped as
$\cup_{i=0}^{\ell} 2^iX_i$. Note that $X_i$ contains at most $4n$ elements as there are $n$ item weights, and each item has  $u_j[i]< 4$ copies. 
Our goal is to determine whether $t\in {\cal S}(\cup_{i=0}^{\ell} 2^iX_i)$ for an arbitrary target value $t$. By Lemma~\ref{lemma:sum-union}, it suffices to determine whether $t\in \sum_{i=0}^{\ell}2^i{\cal S}(X_i)$. We first compute each ${\cal S}(X_i)$. This is equivalent to computing sumset $\sum_{e\in X}\{0,e\}$. According to Lemma~\ref{lemma:sumset}, this can be achieved in $\widetilde{O}(\sigma_1)$ time, where $\sigma_1\leq 4nw_{\max}$. Next,  we apply Lemma~\ref{lemma:k-sum} to determine whether $t\in \sum_{i=0}^{\ell}2^i{\cal S}(X_i)$, which takes $\widetilde{O}(\sigma \ell)$ time, where $\sigma=\sum_{i=0}^{\ell} \texttt{m}_{{\cal S}(X_i)}$, where $\texttt{m}_{{\cal S}(X_i)}$ refers to the largest integer in ${\cal S}(X_i)$, which is bounded by $4nw_{\max}$. Consequently, the overall running time is $\widetilde{O}(\ell^2nw_{\max})=\widetilde{O}(nw_{\max})$. 

\smallskip
\begin{Remark}
We do not require that in the input instance, item weights are distinct.
If, however, there are items of the same weight, we can simply merge them and update $u_j$'s. After preprocessing, we apply our algorithm. Hence, the running time of our algorithm can also be bounded by $\widetilde{O}(n+|W|w_{\max})$ where $W$ stands for the set of distinct weights.
\end{Remark}


\subsection{An $\widetilde{O}(n+w_{\max}^{3/2})$-time Algorithm for Subset Sum}
We present an alternative algorithm for Bounded Subset Sum of running time $\widetilde{O}(n+w_{\max}^{3/2})$. We start with the following observation.  Recall that $W$ is the set of all distinct weights.

\begin{Observation}
    We may assume without loss of generality that $|W|\geq 4c_A w^{1/2}_{\max}\log w_{\max}$, and $2c_A w^{3/2}_{\max}\log w_{\max}\leq t\leq \Sigma_{\cal I}/2$.
\end{Observation}
\begin{proof}
    If $|W|\leq 4c_A w^{1/2}_{\max}\log w_{\max}$, our $\widetilde{O}(n+|W|w_{\max})$ algorithm in the previous subsection already runs in $\widetilde{O}(n+w_{\max}^{3/2})$ time. It remains to consider the case when $|W| \geq 4c_A w^{1/2}_{\max}\log w_{\max}$.

    If $t\geq \Sigma_{\cal I}/2$, we let $t' = \Sigma_{\cal I} - t$.  It is straightforward that the instance $({\cal I}, t)$ is equivalent to the instance $({\cal I}, t')$ and that $t' \leq \Sigma_{\cal I}/2$. If $t < 2c_A w^{3/2}_{\max}\log w_{\max}$, the $\widetilde{O}(n + t)$-time algorithm already runs in $\widetilde{O}(n + w^{3/2}_{\max})$ time. 
\end{proof}


We will use our technique developed for Knapsack to deal with Subset Sum. 
Like Knapsack, this whole subsection is divided into two parts: the structural part and the algorithm part. In the structural part, we will conceptually take the given instance as a special Knapsack instance, and transform it into a 0-1 Subset Sum instance. This allows us to carry over the proximity results from Knapsack. In the algorithm part, we take the original instance and show that it is possible to leverage the proximity results to design our algorithm without involving the instance transformation. 

\subsubsection{The structural part}
In this part, we conceptually view Bounded Subset Sum as a 0-1 Knapsack.
Recall that our Knapsack algorithm starts with a greedy solution. For Subset Sum, every item has the same efficiency, which means we may construct a greedy solution with respect to any order of items. For technical reasons, however, we need to start with a special greedy solution $\veg$. Towards this, we will first define item sets ${\cal I}_2$ and ${\cal I}_3$, use them to define $\veg$, and then use $\veg$ to further define item sets ${\cal I}_1$ and ${\cal I}_4$. These four sets will be treated in a similar way as we treat ${\cal I}_1$ to ${\cal I}_4$ in Knapsack. 

Let $W$ be the set of all item weights. Let
\begin{itemize}
    \item $W_2$ be the set of the largest $2c_A w^{1/2}_{\max} \log w_{\max}$ weights in $W$;
    \item $W_3$ be the set of the smallest $2c_A w^{1/2}_{\max} \log w_{\max}$ weights in $W$.
\end{itemize}
Clearly, $W_2\cap W_3 = \emptyset$. For each $w$ in $W_2$, we put an arbitrary item of weight $w$ into ${\cal I}_2$.  For each $w$ in $W_3$, we put an arbitrary item of weight $w$ into ${\cal I}_3$. The following statements hold.
\begin{enumerate}[label = {\normalfont (\roman*)}]
    \item the total weight of items in ${\cal I}_2$ is at most $2c_A w^{3/2}_{\max}\log w_{\max} \leq t$.

    \item the total weight of items not in ${\cal I}_3$ is at least $\Sigma_{\cal I}/2 \geq t$.
\end{enumerate}
Let $\veg$ be an arbitrary maximal solution such that $\veg$ selects every item in ${\cal I}_2$ but no item in ${\cal I}_3$. In the algorithm part, we shall construct such a solution $\veg$. 

Assuming $\veg$, let ${\cal I}_1 = \veg({\cal I}) \setminus {\cal I}_2$, and let ${\cal I}_4 = \overline{\veg}({\cal I}) \setminus {\cal I}_3$. We relabel the items so that for any $i < j$, every item in ${\cal I}_i$ has smaller index than any item in ${\cal I}_j$.

Recall that Subset Sum is a decision problem. Since we treat it as a special Knapsack problem,  we define its optimal solution as the solution that returns a subset-sum $t^*\leq t$ and is closest to $t$. In particular, if the answer to Subset Sum is ``yes'', then $t^*=t$. Let $\vez$ be the optimal solution with the largest lexicographical order. That is, $\vez$ always prefers items with smaller indices.
We have the following fine-grained proximity between $\veg$ and $\vez$.

\begin{lemma}\label{lemma:subset-proximity}
    $\Delta({\cal I}) = \sum_{i \in {\cal I} }w_i|g_i - z_i| \leq (4c_A + 2c_B)w^{3/2}_{\max}\log w_{\max}$.
\end{lemma}
\begin{proof}
   Note that $\Delta({\cal I}) = \sum_{i = 1}^4 \Delta({\cal I}_i)$. It is obvious that $\Delta({\cal I}_2) + \Delta({\cal I}_3)$ is at most $4c_A w^{3/2}_{\max}\log w_{\max}$ since the total weight of items in ${\cal I}_2 \cup {\cal I}_3$ is bounded by this amount.
    Next we show that $\Delta({\cal I}_1) \leq c_Bw_{\max}^{3/2}$. $\Delta({\cal I}_4)$ can be proved similarly due to symmetry.

    Suppose, for the sake of contradiction, that $\Delta({\cal I}_1) > c_B w_{\max}^{3/2}$. As in the proof of Lemma~\ref{lem:2.5-parition-struct}, we can show that either there is a subset ${\cal D}$ of $\overline{\vez}({\cal I}_1)$ and a subset ${\cal Z}$ of $\vez({\cal I}_2)$ such that ${\cal D}$ and ${\cal Z}$ have the same total weight, or there is a subset ${\cal D}$ of $\overline{\vez}({\cal I}_2)$ and a subset ${\cal A}$ of $\vez({\cal I}_3 \cup {\cal I}_4)$ such that ${\cal D}$ and ${\cal A}$ have the same total weight. Note that in Subset Sum, every item has the same efficiency of $1$. In the former case, $\hat{\vez} = (\vez\setminus {\cal Z}) \cup {\cal D}$ is an optimal solution whose lexicographical order is larger than $\vez$. In the latter case, $\hat{\vez} = (\vez\setminus {\cal A}) \cup {\cal D}$ is an optimal solution whose lexicographical order is larger than $\vez$. Both contradict our choice of $\vez$.
\end{proof}

\subsubsection{The algorithmic part}
Note that the input consists of $n$ pairs $(w_i,u_i)$, meaning that item $i$ has a weight $w_i$ and has $u_i$ copies. 
Despite that $\veg$ is defined on the transformed 0-1 instance, as we can get and sort $W$ in $O(n\log n)$, constructing it can be done efficiently.  
\begin{lemma}
    $\veg$ can be computed in $\widetilde{O}(n)$ time.
\end{lemma}


\begin{lemma}\label{lemma:dense-algo}
  Assuming $\veg$,  with $1-o(1)$ probability we can determine whether there exists $\vez\in\mathbb{Z}^n$, $0\leq z_i\leq u_i$ such that $\sum_{i\in \mathcal{I}}w_i z_i=t$ in $\widetilde{O}(w^{3/2}_{\max})$ time.
\end{lemma}
\begin{proof}
From $\veg$ we define two sets as follows: $G^-=\{(w_i,u_i^-):u_i^-=g_i\}$, $G^+=\{(w_i,u_i^+):u_i^+=u_i-g_i\}$. Intuitively, when changing $\veg$ to $\vez$, $G^-$ and $G^+$ represent copies that may need to be deleted and added, respectively. By Lemma~\ref{lemma:subset-proximity}, the total weight of item copies to be deleted in $G^-$ is at most $\widetilde{O}(w^{3/2}_{\max})$, and the total weight of copies to be added from $G^+$ is also at most $\widetilde{O}(w^{3/2}_{\max})$.

Let $t'=(4c_A + 2c_B)w_{\max}^{3/2}\log w_{\max}$ and $t_0=\sum_{i\in \mathcal{I}}w_i g_i$. Define $\mathcal{S}(G^+,t')=\{y\leq t':y=\sum_{i=1}^nw_ix_i,0\leq x_i\leq u_i^+\}$ and $\mathcal{S}(G^-,t')=\{y\leq t':y=\sum_{i=1}^nw_ix_i,0\leq x_i\leq u_i^-\}$.

We first compute $\mathcal{S}(G^+,t')$ and $\mathcal{S}(G^-,t')$ using Bringmann's algorithm \cite{Bri17}. Note that it is a randomized algorithm. We will compute $S^+\subseteq \mathcal{S}(G^+,t')$ and $S^-\subseteq \mathcal{S}(G^-,t')$ in $\widetilde{O}(t')=\widetilde{O}(w_{\max}^{3/2})$ time such that every number in $\mathcal{S}(G^-,t')$ and $\mathcal{S}(G^+,t')$ is contained in $S^-$ and $S^+$ with probability $1-(n+t')^{-\Omega(1)}$, respectively. Therefore, with probability $1-o(1)$ we get $\mathcal{S}(G^-,t')$ and $\mathcal{S}(G^+,t')$. 
Next, we compute $\mathcal{S}(G^+,t')-\mathcal{S}(G^-,t')$ via FFT in $\widetilde{O}(t')=\widetilde{O}(w_{\max}^{3/2})$ time. Finally, we test whether $t\in t_0+\mathcal{S}(G^+,t')-\mathcal{S}(G^-,t')$, which takes $\widetilde{O}(t')=\widetilde{O}(w_{\max}^{3/2})$ time. Thus, Lemma~\ref{lemma:dense-algo} is proved. \end{proof}

\section*{Remark}
We note that very recently, Bringmann~\cite{bringmann2023knapsack} and Jin~\cite{jin202301} independently improved the running time to $\widetilde{O}(n+w^2_{\max})$ by generalizing our technique.

\newpage
\appendix

\section{Proof of Lemma~\ref{lem:conv-by-group}}\label{Appendix proof lem:conv-by-group}
\lemconvbygroup*
\begin{proof}
Without loss of generality, assume that $U_1\leq ...\leq U_k$.
We construct an optimal solution $\vex$ from $\veg$, which can be interpreted as $\vex = \veg-\vex^-+\vex^+$. Here $\vex^-$ stands for the copies in $\veg$ but not in $\vex$, which must be copies in $\mathcal{I}_{< b}$. Likewise,  $\vex^+$ stands for copies in $\vex$ but not in $\veg$, which must be copies in $\mathcal{I}_{\geq b}$. It is obvious that $\Delta(\mathcal{I}_{< b}) = \sum_{j=1}^k \Delta(\mathcal{I}_{< b}^{W_j})\leq k\cdot  U_k$ and $\Delta(\mathcal{I}_{\geq b}) \leq k\cdot U_k $. 

From now on we shall use the bold symbol $\boldsymbol{x}$ to represent a sequence. We shall compute 3 sequences:
\begin{itemize}
    \item $\boldsymbol{x}^{-} = \langle x^{-}_0,...,x^{-}_{k\cdot U_k}\rangle$, where $x^{-}_{t'}$ is the minimum total profit of copies in $\mathcal{I}_{< b}$ whose total weight is exactly $t'$;
    \item $\boldsymbol{x}^{+} = \langle x^{+}_{ 0},...,x^{+}_{ k\cdot U_k} \rangle$, where $x^{+}_{t'}$ is the maximum total profit of copies in $\mathcal{I}_{\geq b}$ whose total weight is exactly $t'$.
    \item $\boldsymbol{x}^+_{\leq} = \langle x^{+}_{\leq 0},...,x^{+}_{\leq k\cdot U_k} \rangle$, where $x^+_{\leq t'}$ means the maximum value we can obtain by choosing copies from $\mathcal{I}_{\geq b}$ whose total weight does not exceed $t'$. To simply the writing up, we let $x^+_{\leq k\cdot U_k+h}=x^+_{\leq k\cdot U_k}$ for all $1\leq h\leq \Delta$, where $\Delta = t-\sum_{i\in{\mathcal{I}}}w_i g_i$.
\end{itemize}

Suppose we have computed the above three sequences, then Knapsack can be solved as follows: For every $t'\in\{0,...,k\cdot  U_k\}$, we compute $y_{t'} = \sum_{i\in{\mathcal{I}}}v_i g_i-x^-_{t'}+x^+_{\leq(t'+ \Delta)}$, and select $\max_{t'}y_{t'}$.

It remains to compute the three sequences. We show that, $\boldsymbol{x}^{+} = \langle x^{+}_{ 0},...,x^{+}_{ k\cdot U_k} \rangle$ can be computed in $ \widetilde{O}(n+k\sum_{j=1}^{k}|W_j|U_j)$ time. The computation of $\boldsymbol{x}^{-}$ is similar. Moreover, once we have computed $\boldsymbol{x}^{+}$, $\boldsymbol{x}^{+}_{\leq}$ can be computed in a straightforward way using that $x^+_{\leq t'}=\max_{1\leq h\leq t'}x^+_h$.


Now we describe our algorithm for computing $\boldsymbol{x}^{+} = \langle x^{+}_{ 0},...,x^{+}_{ k\cdot U_k} \rangle$.  Towards this, we use the fast structured $(\min, +)$-convolution as a basic operation.
The $(\max, +)$-convolution $\boldsymbol{a}\oplus \boldsymbol{b}$ between two sequence $\boldsymbol{a}$ and $\boldsymbol{b}$ is a sequence $c_0,..., c_{n+m}$ such that for any $i$,
\[
c_i = \max_{0\leq j\leq i}(a_j +b_{i-j}).
\]
Let $\boldsymbol{s}^w=\langle s^{w}_0,...,s^{w}_{U_j}\rangle$, where $s^w_{t'}$ denotes the optimal objective value of copies in $\mathcal{I}_{\geq b}^{w}$ whose total weight is exactly $t'$. Note that copies in $\mathcal{I}_{\geq b}^{w}$ have the same weight $w$. Hence, $s^w_{t'}$ equals the total profit of the $i$ most valuable copies in $\mathcal{I}_{\geq b}^{w}$ when $t' = i\cdot w$ for some $i$, or equals $-\infty$ otherwise. 

Given $\boldsymbol{s}^w$, $\boldsymbol{x}^{+}$ can be computed by iteratively convolving $\boldsymbol{s}^w$, see Algorithm~\ref{alg:lemma3}. Due to the special structure of $\boldsymbol{s}^w$, a single convolution step $\boldsymbol{x}^+ \oplus \boldsymbol{s}^w$ can be performed in linear time by SMAWK algorithm \cite{AKM+87} (also see Lemma 2.2 in \cite{PRW21}).
As $\Delta(\mathcal{I}_{\geq b}^{W_1\cup ...\cup W_{k'}}) \leq \sum_{j=1}^{k'} U_j \leq k'\cdot U_{k'}$, we can truncate $\boldsymbol{x}^+$ after the $k'\cdot U_{k'}$-th entry when we compute the weights in $W_{k'}$. We can compute $ \boldsymbol{x}^{+}$ in time:
$
\widetilde{O}(n+\sum_{j=1}^{k} j\cdot |W_j|\cdot U_j) \leq \widetilde{O}(n+k\sum_{j=1}^{k} |W_j| U_j).
$
\end{proof}
\begin{algorithm}[H]
\caption{The Algorithm to Compute $\boldsymbol{x}^+$ }
\label{alg:lemma3}
    \begin{algorithmic}[1]
    \State sort $W_1,...,W_k$ that $U_1\leq ...\leq U_k$\;
    \State $\boldsymbol{x}^+:=  $  empty sequence\;
    \For{$j= 1,...,k$ }
    \For{$w\in W_j$}
        \State $\boldsymbol{s}^{w} := $ the sequence of maximum value get from $\mathcal{I}_{\geq b}^{w}$\;
        \State $\boldsymbol{x}^+ :=  \boldsymbol{x}^+ \oplus \boldsymbol{s}^{w}$ \Comment{using SMAWK algorithm}\;
        \State Truncate $\boldsymbol{x}^+$ after the $j\cdot U_j$-th entry\;
    \EndFor
  \EndFor
  \State \Return $\boldsymbol{x}^+$\;
  \end{algorithmic}
\end{algorithm}

\section{Proof of Lemma~\ref{lem:hit-b}}\label{Appendix proof lem:hit-b}
\lemhitb*
\begin{proof}
Let $S$ a maximal subset of $A$ such that each integer has multiplicity $w^{1/5}$ in $S$. Note that $|\texttt{supp}_S| \geq c_A w^{2/5} \log^2 w$. What we actually prove is that there is a non-empty subset $S'$ of $S$ such that some subset $B'$ of $B$ have the same total weight as $S$.
We characterize the set of integers that can be hit by $S$. Then we show that $B$ can hit at least one of these integers. 

    Since the items in $S$ has the same multiplicity, we obtain the following two equalities that will be frequently used in the proof.
    \begin{align*}
        &\mu_S\Sigma_{\texttt{supp}_S} = \Sigma_S,\\
        &\mu_S |\texttt{supp}_S| = |S|.
        \label{eq:A-vs-supp}
    \end{align*}
    Note that $\texttt{m}_S \leq w$, $\mu_S = w^{1/5}$, $|\texttt{supp}_S| = |\texttt{supp}_A|\geq c_A w^{2/5} \log^2 w$.  Since $c_A$ is sufficiently large, we have that
    \[
        |S|^2 = \mu^2_S|\texttt{supp}_S|^2 \geq \mu_S^2 \cdot c^2_A w^{4/5} \log^4 w \geq c^2_A\log^4w\cdot \texttt{m}_S\mu_S.
    \]
    The last inequality is due to $\mu_A\geq w^{1/5}$.
    Since $|\texttt{supp}_A| \leq w$ and $\mu_A \leq w$,
    \begin{align*}
        c_\delta &= \Theta\left( \log(2|S|)\log^2(2\mu_S)\right)\\
            &= \Theta\left( \log(\mu_S|\texttt{supp}_S|)\log^2(\mu_S)\right)\\
            & = \Theta\left( \log(w^2)\log^2(w)\right)\\
            & = \Theta(\log^3 w).
    \end{align*}
    When $c_A$ is sufficiently large, $c^2_S\log^4w \geq c_\delta$, so $S$ is $c_\delta$-dense.

    By Theorem~\ref{thm:divisor}, there exists an integer $d$ such that  $S':= S(d)/d$ is $c_{\delta}$-dense and has no $c_\alpha$-almost divisor. And the followings hold.
    \begin{enumerate}[label = {\normalfont (\roman*)}]
        \item $d \leq 4\mu_S\Sigma_S/|S|^2$.
        \item $|S'| \geq 0.75 |S|$.
        \item $\Sigma_{S'}\geq 0.75\Sigma_S/d$.
    \end{enumerate}
    Note that $\mu_{S'} \leq \mu_S$, $\texttt{m}_{S'} \leq w/d$, and $\Sigma_{S'} \leq \Sigma_S / d$. Applying Theorem~\ref{thm:hitrange} on $S'$, we get $S'$ can hit any integer in $[\lambda_{S'}, \Sigma_{S'} - \lambda_{S'}]$ where 
    \[
        \lambda_{S'} = \frac{c_{\lambda} \mu_{S'}\texttt{m}_{S'}\Sigma_{S'}}{|S'|^2} \leq \frac{c_{\lambda} \mu_S w \Sigma_S}{(0.75|S|)^2 d^2} \leq \frac{\min\{c_A, c_B\}}{2} \cdot\frac{w}{d^2} \cdot \frac{\mu_S \Sigma_S}{|S|^2}.
    \]
    The last inequality holds since $c_A$ and $c_B$ are sufficiently large constants. We can conclude that $S$ can hit any multiple of $d$ in $[d\lambda_{S'}, d(\Sigma_{S'} - \lambda_{S'})]$. We also have that the left endpoint of this interval
    \[
        d\lambda_{S'} \leq \frac{c_B}{2} \cdot \frac{w}{d} \cdot \frac{\mu_S\Sigma_S}{|S|^2} = \frac{c_B}{2} \cdot \frac{w}{d} \cdot \frac{\Sigma_{\texttt{supp}_S}}{|\texttt{supp}_S|^2} \leq \frac{c_B}{2} \cdot  \frac{w^2}{|\texttt{supp}_S|}\leq \frac{c_B}{2} \cdot w^{8/5},
    \]
    and that the length of the interval 
    \begin{align*} 
        d (\Sigma_{S'} - 2\lambda_{S'}) 
        \geq &\frac{3\Sigma_S}{4} - c_A\cdot \frac{w \mu_S\Sigma_S}{d|S|^2}\\
        = &\frac{\mu_S\Sigma_S}{|S|^2}(\frac{3|S|^2}{4\mu_S} - c_A\cdot \frac{w}{d}) & &(\text{since $|S| = \mu_S|\texttt{supp}_S|$ and $d \geq 1$})\\
        \geq &\frac{\mu_S\Sigma_S}{|S|^2}(\frac{3|\texttt{supp}_S|^2\mu_S^2}{4\mu_S} - c_A\cdot w) & &(\text{since } |\texttt{supp}_S| \geq c_Aw^{2/5}\log^2 w\text{ and } \mu_S = w^{1/5})\\
        \geq &\frac{\mu_S\Sigma_S}{|S|^2}(\frac{3c^2_Aw}{4} - c_Aw) & &(\text{since $c_A$ is sufficiently large})\\
        \geq &4\cdot \frac{\mu_S\Sigma_S}{|S|^2}\cdot w & &(\text{since } d \leq 4\mu_S\Sigma_S/|S|^2)\\
        \geq &dw.
    \end{align*}

    To complete the proof, it suffices to show that there is a subset $B'$ of $B$ whose sum is a multiple of $d$ and is within the interval $[d\lambda_{S'}, d(\Sigma_{S'} - \lambda_{S'})]$. We claim that as long as $B$ has at least $d$ numbers, there must be a non-empty subset of $B$ whose sum is at most $dw$ and is a multiple of $d$. Assume the claim is true. We can repeatedly extract such subsets from $B$ until $B$ has less than $d$ numbers. Note that the total sum of these subsets is at least
    \[
        \Sigma_B - wd \geq c_B w^{8/5} - w\cdot \frac{4\mu_S\Sigma_S}{|S|^2}\geq c_B w^{8/5} -  \frac{16w\Sigma_{\texttt{supp}_S}}{|\texttt{supp}_S|^2} \geq c_B w^{8/5} - 16w^{8/5} \geq \frac{c_B}{2}w^{8/5}.
    \]
    That is, the total sum of these subset is at least the left endpoint of $[d\lambda_{S'}, d(\Sigma_{S'} - \lambda_{S'})]$. Also note that the sum of each subset is at most $dw$, which does not exceed the length of the interval. As a result, there must be a collection of subsets whose total sum is within the interval.  Since the sum of each subset is a multiple of $d$, so is the any collection of these subset.
\end{proof}

\section{Proof of Lemma~\ref{lem:2.4-partition-b}}
\lempartitionb*
\begin{proof}
Recall that we label the items in decreasing order of efficiency. Without loss of generality, we assume that for any two items $i < j < b$, if $i$ and $j$ have the same efficiency, then $w_i \leq w_j$, and that for any two items $b < i < j$, if $i$ and $j$ have the same efficiency, then $w_i \geq w_j$.  For a set ${\cal I}'$ of items, we say a weight $w$ is \emph{frequent} in ${\cal I}'$ if ${\cal I}'$ contains at least $2w^{1/5}_{\max}$ items with weight $w$.

\paragraph{Defining $W^+$ via Partition of ${\cal I}^{W^*}$.} $W^+$ will be defined via a partition of ${\cal I}^{W^*}$. We partition them into four subsets $({\cal I}_1, {\cal I}_2, {\cal I}_3, {\cal I}_4)$ as follows. For $i < b$, let ${\cal I}_{[i, b)}$ be the set of items $\{i ,\ldots, b-1\}$ whose weight is in $W^*$. Let $i^*$ be the minimum index $i$ such that exactly $2c_Aw^{2/5}_{\max}\log^2 w_{\max}$ weights are frequent in ${\cal I}_{[i, b)}$. Let ${\cal I}_2 = {\cal I}_{[i^*, b)}$, and let ${\cal I}_1 = {\cal I}_{< i^*}$. When no such $i^*$ exists, let ${\cal I}_2 = {\cal I}_{<b}$, and let ${\cal I}_1 = \emptyset$. ${\cal I}_3$ and ${\cal I}_4$ are defined similarly as follows. For any $j \geq b$, define ${\cal I}_{[b,j]}$ to be the set of items $\{b, \ldots, j\}$ whose weight is in $W^*$. Let $j^*$ be the maximum index $j$ such that exactly $2c_A w^{2/5}_{\max}\log^2 w_{\max}$ weights are frequent. Let ${\cal I}_3 = {\cal I}_{[b,j^*]}$, and let ${\cal I}_4 = {\cal I}_{> j^*}$. When no such $j^*$ exists, let ${\cal I}_3 = {\cal I}_{\geq b}$, and let ${\cal I}_4 = \emptyset$. 

$W^+$ is define to be the set of weights that are frequent in ${\cal I}_2$ or ${\cal I}_3$, and $\overline{W^+} = W^*\setminus W^+$.

\paragraph{Verifying Properties.} It is straightforward that $|W^+| \leq 4c_A w_{\max}^{2/5}\log^2 w_{\max}$. It is left to show $\Delta({\cal I}^{\overline{W^+}}) \leq 8c_Bw^{9/5}_{\max}$. We first partition ${\cal I}^{\overline{W^+}}$ into ${\cal I}^{\overline{W^+}}_k = {\cal I}^{\overline{W^+}} \cap {\cal I}_k$ for $k \in \{1,2,3,4\}$. Next we show $\Delta({\cal I}^{\overline{W^+}}_k) \leq 2c_B w^{9/5}_{\max}$ for $k \in \{1,2\}$. 

Consider ${\cal I}^{\overline{W^+}}_2$ first. For any $w \in \overline{W^+}$, $w$ is not frequent in ${\cal I}_2$. That is, $|{\cal I}_2 \cap {\cal I}^w| < w^{1/5}_{\max}$. We have
    \[
        \Delta({\cal I}^{\overline{W^+}}_2) \leq \sum_{i \in {\cal I}^{\overline{W^+}}_2} w_i = \sum_{w \in \overline{W^+}} w|{\cal I}_2 \cap {\cal I}^w| \leq |\overline{W^+}|w_{\max}w^{1/5}_{\max} = w^{9/5}_{\max}.
    \]

Now consider ${\cal I}^{\overline{W^+}}_1$. Suppose, for the sake of contradiction, that $\Delta({\cal I}^{\overline{W^+}}_1) > 2c_Bw^{9/5}_{\max} > 2c_Bw^{8/5}_{\max}$. That is, we delete a large volume of items in ${\cal I}^{\overline{W^+}}_1$. Clearly, ${\cal I}^{\overline{W^+}}_1$ is not empty, so there are exactly $2c_Aw^{2/5}_{\max}\log^2 w_{\max}$ weights that are frequent in ${\cal I}_2$. Recall that, by frequent, we mean $|{\cal I}^w \cap {\cal I}_2| \geq 2w^{1/5}_{\max}$. Then at least one of the following two cases are true.
    \begin{enumerate}[label={(\roman*)}]
        \item there are at least $c_Aw^{2/5}_{\max}\log^2 w_{\max}$ weights $w$ such that  $|{\cal I}^w \cap \vez({\cal I}_2)| \geq w^{1/5}_{\max}$.

        \item there are at least $c_Aw^{2/5}_{\max}\log^2 w_{\max}$ weights $w$ such that  $|{\cal I}^w \cap \overline{\vez}({\cal I}_2)| \geq w^{1/5}_{\max}$.
    \end{enumerate}

\paragraph{Case (i).} $\overline{\vez}({\cal I}^{\overline{W^+}}_1)$ is a set of deleted items whose total weight is $\Delta({\cal I}^{\overline{W^+}}_1) > 2c_Bw^{8/5}_{\max}$. $\vez({\cal I}_2)$ is a set of items remaining in $\vez$, and there are at least $c_Aw^{2/5}_{\max}\log^2 w_{\max}$ weights $w$ such that  $|{\cal I}^w \cap \vez({\cal I}_2)| \geq w^{1/5}_{\max}$. One can verify that the weight multisets of $\overline{\vez}({\cal I}^{\overline{W^+}}_1)$ and $\vez({\cal I}_2)$ satisfy the conditions of Lemma~\ref{lem:hit-b}. By Lemma~\ref{lem:hit-b}, there is a subset ${\cal D}$ of $\overline{\vez}({\cal I}^{\overline{W^+}}_1)$ and a subset ${\cal Z}$ of $\vez({\cal I}_2)$ such that ${\cal D}$ and ${\cal Z}$ have the same total weight. Then $\hat{\vez} = (\vez \setminus {\cal Z}) \cup {\cal D}$ would be a better solution than $\vez$. Contradiction.

\paragraph{Case (ii).} there are at least $c_Aw^{2/5}_{\max}\log^2 w_{\max}$ weights $w$ such that  $|{\cal I}^w \cap \overline{\vez}({\cal I}_2)| \geq w^{1/5}_{\max}$. $\Delta({\cal I}^{\overline{W^+}}_1) > 2c_Bw^{8/5}_{\max}$ implies that 
    \[
        \Delta({\cal I}_{\geq b}) \geq \Delta({\cal I}^{\overline{W^+}}_1) - w_{\max} \geq c_Bw^{8/5}_{\max}.
    \]
    Note that the total weight of $\vez({\cal I}_{\geq b})$ is exactly $\Delta({\cal I}_{\geq b})$. One can verify that the weight multisets of $\overline{\vez}({\cal I}_2)$ and $\vez({\cal I}_{\geq b})$ satisfy the condition of Lemma~\ref{lem:hit-b}. By Lemma~\ref{lem:hit-b}, there is a subset ${\cal D}$ of $\overline{\vez}({\cal I}_2)$ and a subset ${\cal A}$ of $\vez({\cal I}_{\geq b})$ such that ${\cal D}$ and ${\cal A}$ have the same total weight. Then $\hat{\vez} = (\vez \setminus {\cal A}) \cup {\cal D}$ would be a better solution than $\vez$. Contradiction.

   Due to symmetry, it can be similarly proved that $\Delta({\cal I}^{\overline{W^+}}_k) \leq 2c_B w^{9/5}_{\max}$ for $k \in \{3,4\}$. We provide a full proof for completeness. Consider ${\cal I}^{\overline{W^+}}_3$. For any $w \in \overline{W^+}$, $w$ is not frequent in ${\cal I}_2$. That is, $|{\cal I}_3 \cap {\cal I}^w| < w^{1/5}_{\max}$.
    \[
        \Delta({\cal I}^{\overline{W^+}}_3) \leq \sum_{i \in {\cal I}^{\overline{W^+}}_3} w_i = \sum_{w \in \overline{W^+}} w|{\cal I}_3 \cap {\cal I}^w| \leq |\overline{W^+}|w_{\max}w^{1/5}_{\max} = w^{9/5}_{\max}.
    \]
   Consider ${\cal I}^{\overline{W^+}}_4$. Suppose, for the sake of contradiction, that $\Delta({\cal I}^{\overline{W^+}}_4) > 2c_Bw^{9/5}_{\max}>2c_Bw^{8/5}_{\max}$. That is, we add a large volume of items from ${\cal I}^{\overline{W^+}}_4$ when obtaining $\vez$ from $\veg$. Clearly, ${\cal I}^{\overline{W^+}}_4$ is not empty, so there are exactly $2c_Aw^{2/5}_{\max}\log^2 w_{\max}$ weights that are frequent in ${\cal I}_3$. Recall that, by frequent, we mean $|{\cal I}^w \cap {\cal I}_3| \geq 2w^{1/5}_{\max}$. Then at least one of the following two cases are true.
    \begin{enumerate}[resume*]
        \item there are at least $c_Aw^{2/5}_{\max}\log^2 w_{\max}$ weights $w$ such that  $|{\cal I}^w \cap \overline{\vez}({\cal I}_3)| \geq w^{1/5}_{\max}$.

        \item there are at least $c_Aw^{2/5}_{\max}\log^2 w_{\max}$ weights $w$ such that  $|{\cal I}^w \cap \vez({\cal I}_3)| \geq w^{1/5}_{\max}$.
    \end{enumerate}

    \paragraph{Case (iii).} $\vez({\cal I}^{\overline{W^+}}_4)$ has a total weight of $\Delta({\cal I}^{\overline{W^+}}_4) > 2c_Bw^{8/5}_{\max}$. For $\overline{\vez}({\cal I}_3)$, there are at least $c_Aw^{2/5}_{\max}\log^2 w_{\max}$ weights $w$ such that  $|{\cal I}^w \cap \overline{\vez}({\cal I}_3)| \geq w^{1/5}_{\max}$. One can verify that the weight multisets of $\vez({\cal I}^{\overline{W^+}}_4)$ and $\overline{\vez}({\cal I}_3)$ satisfy the conditions of Lemma~\ref{lem:hit-b}. By Lemma~\ref{lem:hit-b}, there is a subset ${\cal Z}$ of $\vez({\cal I}^{\overline{W^+}}_4)$ and a subset ${\cal D}$ of $\overline{\vez}({\cal I}_3)$ such that ${\cal D}$ and ${\cal Z}$ have the same total weight. Then $\hat{\vez} = (\vez \setminus {\cal Z}) \cup {\cal D}$ would be a better solution than $\vez$. Contradiction.

    \paragraph{Case (iv).} there are at least $c_Aw^{2/5}_{\max}\log^2 w_{\max}$ weights $w$ such that  $|{\cal I}^w \cap \vez({\cal I}_3)| \geq w^{1/5}_{\max}$. $\Delta({\cal I}^{\overline{W^+}}_4) > 2c_Bw^{8/5}_{\max}$ implies that 
    \[
        \Delta({\cal I}_{< b}) \geq \Delta({\cal I}^{\overline{W^+}}_4) - w_{\max} \geq c_Bw^{8/5}_{\max}.
    \]
    Note that the total weight of $\overline{\vez}({\cal I}_{< b})$ is exactly $\Delta({\cal I}_{< b})$. One can verify that the weight multisets of $\vez({\cal I}_3)$ and $\overline{\vez}({\cal I}_{< b})$ satisfy the conditions of Lemma~\ref{lem:hit-b}. By Lemma~\ref{lem:hit-b}, there is a subset ${\cal A}$ of $\vez({\cal I}_3)$ and a subset ${\cal D}$ of $\overline{\vez}({\cal I}_{< b})$ such that ${\cal D}$ and ${\cal A}$ have the same total weight. Then $\hat{\vez} = (\vez \setminus {\cal A}) \cup {\cal D}$ would be a better solution than $\vez$. Contradiction.
\end{proof}

\section{Proof of Lemma~\ref{lemma:sum-union}}
\lemsumunion*
\begin{proof}
    It is obvious that ${\cal S}(X)\supseteq  \sum_{i=0}^{\ell}2^i{\cal S}(X_i)$. We prove below that ${\cal S}(X)\subseteq  \sum_{i=0}^{\ell}2^i{\cal S}(X_i)$, and Lemma~\ref{lemma:sum-union} follows. Take an arbitrary $t'\in {\cal S}(X)$, by definition there exist $\eta_j\leq u_j=\sum_{i=0}^{\ell}u_j[i]\cdot 2^i$ for all $1\leq j\leq n$ such that
    \[t'=\sum_j w_j\eta_j.\]

We claim that 
\begin{Claim}\label{claim:1}
If $\eta_j\leq u_j=\sum_{i=0}^{{\ell}_j}u_j[i]\cdot 2^i$, then there exist $\eta_j[i]$'s such that $\eta_j=\sum_{i=0}^{{\ell}_j}\eta_j[i]\cdot 2^i$ and $\eta_j[i]\leq u_j[i]$ for all $0\leq i\leq {\ell}_j$. 
\end{Claim}
Suppose the claim is true, then since $X_i$ contains $u_j[i]$ copies of weight $w_j$, we know that $\sum_j\eta_j[i]w_j\in {\cal S}(X_i)$. Thus, $\sum_j\eta_j[i]\cdot 2^iw_j\in 2^i{\cal S}(X_i)$, and consequently $t'=\sum_i\sum_j\eta_j[i]\cdot 2^iw_j\in \sum_i2^i{\cal S}(X_i)$, and Lemma~\ref{lemma:sum-union} is proved.

It remains to prove Claim~\ref{claim:1}. We prove by induction on ${\ell}_j$. Claim~\ref{claim:1} is obviously true for ${\ell}_j=0$. Suppose it is true for ${\ell}_j=h-1$, we prove that it is also true for ${\ell}_j=h$. Towards this, let $r_t\in \{0,1\}$ be the residue of $t$ modulo $2$. Recall that $u_j[i]\in [2,4)$ for $0\leq i\leq {\ell}_j-1$. There are two possibilities: (i). If $r_t\leq u_j[0]-2$, then we let $\eta_j[0]=r_t+2$; (ii). If  $u_j[0]-2<r_t\leq u_j[0]$, then we let $\eta_j[0]=r_t$. It is clear that $\eta_j[0]\leq u_j[0]$ is always true. Meanwhile, we also have $\eta_j[0]+2>u_j[0]$. 

By the fact that $\eta_j[0]\equiv\eta_j\pmod 2$ and $u_j[0]\equiv u_j \pmod 2$, $\eta_j-\eta_j[0]$ and $u_j-u_j[0]$ are multiples of $2$. Using that $\eta_j\leq u_j$ and $\eta_j[0]+2>u_j[0]$, we have
\[\eta_j-\eta_j[0]<u_j-u_j[0]+2.\]
Thus,
\[\frac{\eta_j-\eta_j[0]}{2}<\frac{u_j-u_j[0]}{2}+1.\]
Since $\frac{\eta_j-\eta_j[0]}{2}$ and $\frac{u_j-u_j[0]}{2}$ are both integers, we have that 
\[\frac{\eta_j-\eta_j[0]}{2}\leq \frac{u_j-u_j[0]}{2}=\sum_{i=0}^{\ell-1}u_j[i+1]\cdot 2^i.\]
Using the induction hypothesis, there exists $\bar{\eta}_j[i]$ such that
$\frac{\eta_j-\eta_j[0]}{2}=\sum_{i=0}^{\ell-1}\bar{\eta}_j[i]\cdot 2^i$ where $\bar{\eta}_j[i]\leq u_j[i+1]$. Recall that $\eta_j[0]\leq u_j[0]$. Consequently, Claim~\ref{claim:1} is true.
\end{proof}

\bibliographystyle{plainurl}
\bibliography{ref}

\end{document}